\documentclass[journal]{IEEEtran}
 
\usepackage{xkeyval}
\usepackage{algorithm}
\usepackage{algorithmicx}
\usepackage{algpseudocode}
\usepackage{amsfonts}
\usepackage{bm}
\usepackage[dvipsnames]{xcolor}
\usepackage[normalem]{ulem} 
\usepackage{import}
\usepackage{tikz}
\usepackage{xcolor}
\usetikzlibrary{arrows.meta,calc,shapes.geometric}
\definecolor{cFroz}{RGB}{90,0,90}
\definecolor{cICW}{RGB}{106,61,154}
\definecolor{cGrad}{RGB}{217,95,2}
\definecolor{cBest}{RGB}{27,158,119}
\definecolor{cCost}{RGB}{215,35,35}
\usepackage{footnote}

\usepackage{booktabs,adjustbox,colortbl, rotating}
\usepackage[dvipsnames]{xcolor}
\newif\ifshowred
\showredtrue 

\definecolor{darkpurple}{RGB}{255,0,255}
\newif\ifshowpurple
\showpurpletrue 

\usepackage{hyperref}
\usepackage{ragged2e} 
\usepackage{nopageno}
\usepackage{amsthm,amsmath,amssymb}
\allowdisplaybreaks

\newtheorem{proposition}{Proposition}
\newtheorem{remark}{Remark}

\hypersetup{
    colorlinks=true,
    linkcolor=blue,
    filecolor=magenta,      
    urlcolor=blue,
    pdftitle={},
    pdfpagemode=FullScreen,
    citecolor = blue,
    }
\usepackage{multirow}
\usepackage{graphicx, pifont} 
\usetikzlibrary{arrows.meta}
\usetikzlibrary{positioning,fit,backgrounds}
\usetikzlibrary{decorations.pathmorphing,patterns}
\usepackage{tikz}
\usepackage{tikz-network}
 \usepackage{xcolor}
 \usepackage{array}
\usepackage{booktabs}
\usepackage{makecell}
\usepackage{pifont}

\usepackage{ifdraft}
\algnewcommand\Output{\item[\textbf{Output:}]}
\usepackage{mathtools}

\DeclareMathOperator*{\argmin}{arg\,min}

\usepackage[noadjust]{cite}
\usepackage{placeins} 

\definecolor{ok}{RGB}{34,139,34}
\definecolor{bad}{RGB}{200,30,30}
\definecolor{mid}{RGB}{210,140,0}

\begin{document}



\title{Gradient-Free Topology Adaptation for Power Flow Surrogates via In-Context Whitening}

\author{ Ayushi Jolotia and Parikshit Pareek$^{\dagger}$ \vspace{-10mm}
\thanks{\noindent $^\dagger$Corresponding Author;$\,$ \\ 
Authors are with the Department of Electrical Engineering, Indian Institute of Technology Roorkee. 
\textit{\{ayushi\_j;pareek\}@ee.iitr.ac.in}; The authors acknowledge the funding support provided by the ANRF PM Early Career Research Grant (ANRF/ECRG/2024/001962/ENS) and the IIT Roorkee Faculty Initiation Grant (IITR/SRIC/1431/FIG-101078). }}

\maketitle

\begin{abstract}
Machine-learned surrogates for the AC power flow (ACPF) problem amortize the cost of repeated solves on a fixed network, but lose one to two orders of magnitude of accuracy when a line outage changes the topology. This degradation is an operator shift. The altered admittance matrix changes the input-to-output map, so identical inputs yield a different output distribution. Existing methods correct this with target-topology data and per-topology gradient steps. We ask whether the correction can instead be made statistical and gradient-free. We propose In-Context Whitening (ICW), which trains an ACPF surrogate in an output space whitened by the base topology's first two moments, and adapts it to an unseen N-1 or N-2 topology by re-estimating that whitening from a few hundred solved cases on the new topology. This adaptation is gradient-free, weight-free, and architecture-agnostic. We prove that among affine whiteners the unique choice that preserves the coordinate-wise semantics of the physical output vector is ZCA whitening, so within efficient invertible corrections, two moments are sufficient. Across the IEEE 30-, 118-, and 300-bus systems under N-1 and N-2 contingencies, ICW reduces overall error by 6$\times$ to 28$\times$ over frozen surrogates (up to 54$\times$ per-quantity under N-2) and cuts worst-bus power-balance mismatch by up to 30$\times$, with consistent
gains across three backbones. At deployment scale it matches or beats gradient-based adaptation in accuracy while adapting 21$\times$ to 34$\times$ faster, with a cost that parallelizes on commodity CPU cores rather than requiring one GPU per contingency.
\end{abstract}

\section{Introduction}
The AC power flow (ACPF) problem is the core computation in grid operations: given generation and demand, it solves the remaining steady-state variables (e.g., bus voltages, generator reactive power, and slack-bus injection) \cite{wood2013power}. Further, as operators must assess many load/topology conditions, especially under uncertainty in demand, renewables, and network topology, ACPF is run over large scenario ensembles as probabilistic power flow (PPF) \cite{ppf,ppf2}. To address these computational challenges, data-driven surrogate models are an alternative. Trained on Newton--Raphson load flow (NRLF) \cite{wood2013power} solved instances, an ML model can predict ACPF solutions in a forward pass, amortizing the cost of repeated NRLF evaluations across operating conditions on the same topology. These models can also provide initializations or fallback predictions when the iterative solver struggles. Prior work spans feedforward networks \cite{yaniv2023, fikri2018}, graph neural networks (GNNs)\cite{donon2020, lin2024}, physics-informed networks \cite{hu2021}, data-driven linearisations \cite{liu2019}, and Gaussian-process-based formulations \cite{pareek2022}. On their training topology of power network, they achieve acceptable accuracy while running orders of magnitude faster than NRLF. See the community resource \cite{mlopf_wiki} ML-based power-flow surrogates.

However, surrogates trained on a fixed topology can fail as soon as the topology changes: outages or switching modify the admittance matrix and therefore the input--output operator. Even with the same inputs, the solution distribution shifts, with errors often increasing by 1--2 orders of magnitude on N-1 contingencies and worsening on N-2 \cite{nazari2025,ly2025}. A natural fix is to provide the network graph explicitly (e.g., GNN-based PowerFlowNet \cite{lin2024}), but this alone does not eliminate the shift because ACPF solutions depend on the interaction between injections and the altered network, and message-passing weights learned on training load--graph interactions must extrapolate on unseen topologies \cite{nazari2025,ly2025}. This brittleness makes topology adaptation unavoidable. To recover accuracy, surrogates typically need \textit{target-topology data} and \textit{gradient-based updates} (retraining or fine-tuning). Existing approaches follow this pattern, e.g., per-topology auxiliary fine-tuning like UPITL \cite{nazari2025}, few-shot fine-tuning from a meta-learned initialisation \cite{chen2022meta,jia2024,zhou2022deepopfft}, or training separate models for clusters of similar topologies like MMNP \cite{ly2025} or Modular GP~\cite{caetano2026modulargp}. The bottleneck is not inference but the per-topology data generation and optimisation, so reducing this adaptation time (data + compute) remains the key challenge.

In this paper, rather than proposing a new architecture or relying on retraining, we ask two questions. First, \textit{can contingency adaptation be recast as statistical adaptation, by training the surrogate on the base topology in an output space that suppresses topology-specific distributional structure?} Second, \textit{can the contingency-specific statistics then be reinstated at inference from a small set of NR-solved samples, without additional gradient steps or architectural changes, enabling gradient-free, in-context adaptation to a new topology?}

We answer both questions by operationalizing statistical adaptation as distribution alignment in the surrogate's output space. A guiding principle is that, in practice, two distributions tend to be closer when their low-order moments match. By the domain adaptation bound of Ben-David et al.~\cite{bendavid2010}, a smaller distance between the base topology's and the contingency topology's output distributions implies a smaller transfer error for the surrogate. Moreover, matching low-order moments is a standard and effective way to reduce this distance in domain adaptation \cite{sun2016coral,li2018adabn}, though these works act on internal feature representations, whereas we act directly on the surrogate's output space. To make this output space practical, the transform must be efficient and invertible, since the benefit of replacing NRLF with a surrogate $f_w$ is lost if the transformation around $f_w$ is heavier than the forward pass or lacks a closed-form inverse. Finally, once the surrogate is trained in the transformed output space, its predictions must be mapped back to physical quantities at inference with minimal time overhead and without additional gradient steps or architectural changes to the surrogate backbone.

We propose \textit{In-Context Whitening} (ICW), a lightweight, invertible output-space transform that aligns contingency outputs by matching low-order moments. The method collects the required statistics estimated from a small set of solved ACPF cases on each new topology, i.e., \textit{in context}, rather than frozen at training time.
Concretely, ICW carries a mean vector and a covariance matrix, matches these first two moments exactly, and inverts through a single matrix--vector product with one matrix inversion per topology and not per sample.
At inference on a new contingency, the moments required by this whitening transformation are re-estimated from its context set, yielding a topology-specific $T$ and $T^{-1}$ without any gradient steps. Accordingly, an ICW-based surrogate is trained to predict in the whitened space, learning $f_w:\mathbf{x}\to\mathbf{z}$ instead of $f_w:\mathbf{x}\to\mathbf{y}$, where $\mathbf{z}=T(\mathbf{y})$.
At test time, the surrogate output $\widehat{\mathbf{z}}$ is mapped back to physical quantities via $\widehat{\mathbf{y}} = T^{-1}(\widehat{\mathbf{z}})$, separating ICW from inference-time de-normalization with fixed training statistics. The main contributions of this paper are summarized as follows:

\begin{itemize}

\item \textbf{Gradient-free topology adaptation.} We propose In-Context Whitening (ICW), which trains an ACPF surrogate on base-topology data in a whitened output space and adapts it to an unseen N-1 or N-2 topology by re-estimating the whitening statistics \textit{in context} from a small set of NR-solved samples (a few hundred cases). This per-topology adaptation requires no retraining, gradient steps, or architectural changes, and adds only a lightweight affine transform (one matrix inversion per topology). Its adaptation cost is therefore context collection alone, incurred once per contingency, with no gradient budget to tune and nothing to re-optimise as outages accumulate.

\item \textbf{Efficient, coordinate-preserving whitening with two-moment adaptation.} We show that practical requirements of efficiency and invertibility naturally lead to affine transformations and thereby restrict statistical adaptation to the first two moments (Section~\ref{subsec:affine}). Among affine whiteners, we prove that the unique choice minimizing distortion of the physical coordinates is ZCA whitening (Proposition~\ref{prop:zca}). Empirically, a lossless third-moment correction yields no further accuracy gain beyond matching mean and covariance.

\item \textbf{Consistent gains across systems, contingencies, and backbone architectures.} Across the IEEE 30-, 118-, and 300-bus systems under N-1 and N-2 contingencies and three backbones (MLP, Transformer, and PFGNN~\cite{lin2024}), ICW matches the accuracy of context-using adaptation baselines while avoiding per-topology gradient steps, and reduces overall contingency error by $6\times$ to $28\times$ relative to frozen, non-adaptive surrogates (up to $54\times$ per-quantity under N-2).
\end{itemize}

\begin{figure}[h]
\centering
\begin{tikzpicture}[
  >=Stealth, font=\small,
  ax/.style={->, gray!55, thick},
]
\draw[ax] (0,0) -- (7.4,0);
\draw[ax] (0,0) -- (0,3.5);
\node[rotate=90, anchor=south] at (-0.2,2.0) {Prediction error\,$\downarrow$};
\node[anchor=north] at (4.2,-0.1) {Adaptation cost per topology\,$\rightarrow$};

\draw[cBest, dashed, line width=0.9pt] (0.4,0.45) -- (7.3,0.45);
\node[anchor=north, cBest, font=\footnotesize] at (3.2,0.5) {Accuracy floor};

\draw[cFroz!80, dotted, line width=0.9pt] (1.6,0.9) -- (1.6,3.2);
\node[anchor=east, cFroz, font=\footnotesize, align=right] at (1.5,2.0)
  {Same Cost\\$T_{\mathrm{data}}$};

\draw[cGrad, line width=1.1pt]
  (1.6,3.2) .. controls (2.2,2.05) and (4.3,0.82) .. (7.0,0.6);
\node[cGrad, font=\footnotesize, align=center] at (4.2,2)
  {Fine-tune / Scratch\\(gradient steps)};
\filldraw[cBest] (7.0,0.6) circle (2.4pt);
\node[anchor=west, cBest, font=\footnotesize, align=center] at (5.3,1.1)
  {Trained per topology\\(best accuracy)};

\filldraw[cFroz] (0.9,3.1) circle (2.4pt);
\node[anchor=south, cFroz, font=\footnotesize] at (0.9,3.2) {frozen};

\node[star, star points=5, star point ratio=2.3, minimum size=13pt, inner sep=0,
      fill=cICW, draw=cICW] at (1.6,0.9) {};
\node[anchor=west, cICW, font=\footnotesize, align=left] at (1.92,0.9)
  {\textbf{Proposed ICW}\\ (no gradient step)};

\node[draw, rounded corners, anchor=north west, inner sep=2pt,
      text width=8cm, font=\footnotesize, align=left] at (-0.45,-0.7)
  {\textbf{Per new topology}\\[2pt]
   \textbf{ICW}:\ $C=T_{\mathrm{data}}$ {\scriptsize(no grad step)}$\quad$
   \textbf{Grad}:\ $C=T_{\mathrm{data}}+\textcolor{cCost}{T_{\mathrm{grad}}}$};

\node[draw, rounded corners, anchor=north west, inner sep=2pt,
      text width=8cm, font=\footnotesize, align=left] at (-0.45,-1.5)
  {\textbf{Across $K$ contingencies}\\[2pt]
   \textbf{ICW}:\ $K\,T_{\mathrm{data}}$ $\quad$
   \textbf{Grad}:\ $K\big(T_{\mathrm{data}}+\textcolor{cCost}{T_{\mathrm{grad}}}\big)$};

\end{tikzpicture}
    \caption{Cost-accuracy position of ICW versus gradient-based adaptation on a new topology. Axes: per-topology adaptation cost (horizontal) and prediction error (vertical). Gradient-based methods (fine-tune, scratch) trace a curve from the frozen error down toward the accuracy floor as gradient steps accumulate, at cost $T_\mathrm{data}+T_\mathrm{grad}$. ICW (star) collects the same context set but takes no gradient step, reaching near-floor accuracy at cost $T_\mathrm{data}$ alone. Across $K$ contingencies the gradient cost compounds as $K(T_\mathrm{data}+T_\mathrm{grad})$ while ICW's stays at $K\,T_\mathrm{data}$.}
\label{fig:cost-position}
\end{figure}
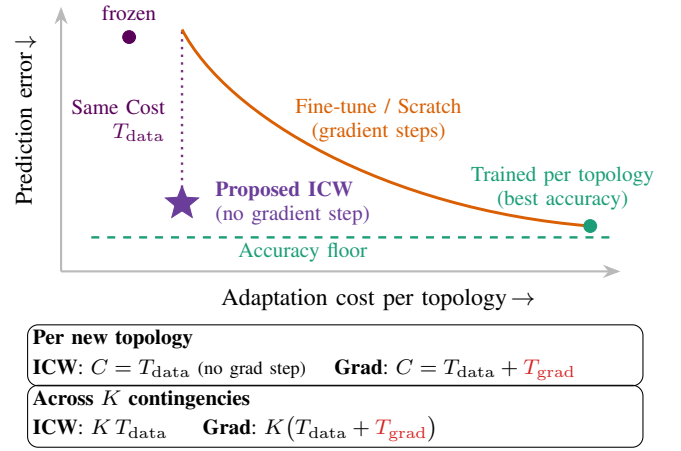

\section{Motivation and Problem Formulation}\label{sec:background}

The AC power flow (ACPF) problem can be posed as solving a nonlinear system of equations $F(\mathbf{x},\mathbf{y};\mathbf{Y}_{\mathrm{bus}})=0$, where $\mathbf{x}$ collects specified operating quantities (e.g., active/reactive demands, generator voltage setpoints, non-slack generator active-power setpoints, and the slack-bus angle reference) and $\mathbf{y}$ collects the remaining steady-state variables to be determined (e.g., bus voltage magnitudes/angles and generator power injections). For a fixed network topology represented by the bus admittance matrix $\mathbf{Y}_{\mathrm{bus}}$, these equations implicitly define an input--output operator $\mathcal{F}_{\mathbf{Y}_{\mathrm{bus}}}:\mathbf{x}\mapsto\mathbf{y}$.

Consider a power network with bus set $\mathcal{N}$ and transmission-line set $\mathcal{E}$. The complex voltage at bus $j\in\mathcal{N}$ is $v_j = V_j\angle\theta_j$, where $V_j$ and $\theta_j$ denote voltage magnitude and phase angle. Let $p_j$ and $q_j$ denote the net active and reactive power injections at bus $j$. For a given topology encoded by $\mathbf{Y}_{\mathrm{bus}} \in \mathbb{C}^{|\mathcal{N}| \times |\mathcal{N}|}$, the ACPF equations are
\begin{equation}
  p_j + \mathrm{j}\,q_j
  \;=\;
  v_j \sum_{k \in \mathcal{N}} \overline{Y}_{jk}\, \overline{v}_k,
  \qquad j \in \mathcal{N}.
  \label{eq:acpf}
\end{equation}
Here, $Y_{jk}$ is the $(j,k)$-th entry of $\mathbf{Y}_{\mathrm{bus}}$ and $\overline{(\cdot)}$ denotes complex conjugation. 
The ACPF task is to solve for $\mathbf{y}$ satisfying~\eqref{eq:acpf}. Because the equations are nonlinear, solutions are typically computed with an iterative method such as NRLF. While a single deterministic solve is modest for small-to-medium networks, the cost grows substantially in uncertainty-aware studies that require thousands of solves across load realizations, generation scenarios, or contingency conditions.

To solve this, ACPF surrogate learning instead trains a parameterized model on pre-solved instances represented as input--output pairs $(\mathbf{x},\mathbf{y})$. We use the input partition
$ \mathbf{x}
  =
  \bigl[
  \mathbf{p}_{l},\;
  \mathbf{q}_{l},\;
  \mathbf{v}_{g},\;
  \mathbf{p}^{\mathrm{ns}}_{g},\;
  \boldsymbol{\theta}^{\mathrm{s}}
  \bigr],$
where $\mathbf{p}_{l}$ and $\mathbf{q}_{l}$ are active/reactive load demands, $\mathbf{v}_{g}$ are generator voltage setpoints, $\mathbf{p}^{\mathrm{ns}}_{g}$ are active-power setpoints of non-slack generators, and $\boldsymbol{\theta}^{\mathrm{s}}$ is the slack-bus angle reference. The corresponding output is
$  \mathbf{y}
  =
  \bigl[
  \mathbf{p}_{g}^{\mathrm{s}},\;
  \mathbf{q}_{g},\;
  \mathbf{v}_{l},\;
  \boldsymbol{\theta}^{\mathrm{ns}}
  \bigr],$
where $\mathbf{p}_{g}^{\mathrm{s}}$ is the slack-bus active power injection, $\mathbf{q}_{g}$ are reactive power outputs of all generators, $\mathbf{v}_{l}$ are voltage magnitudes at load buses, and $\boldsymbol{\theta}^{\mathrm{ns}}$ are phase angles at all non-slack buses. Superscripts $\mathrm{s}$ and $\mathrm{ns}$ denote slack and non-slack quantities, respectively.

A surrogate model $f_w$ with parameters $w$ is then trained to predict $\mathbf{y}$ from $\mathbf{x}$. Once trained, it produces an estimate in a forward pass,
\begin{align}
  \widehat{\mathbf{y}} = f_w(\mathbf{x}).
  \label{eq:pf_functional}
\end{align}

In the next subsection, we show that when the network topology changes (i.e., $\mathbf{Y}_{\mathrm{bus}}$ changes), the conditional mapping from $\mathbf{x}$ to $\mathbf{y}$, and hence the distribution of outputs, shifts, causing a surrogate trained on the base topology to degrade.

\subsection{Operator Shift Under Topology Changes}
\label{subsec:operator_shift}

The ACPF equations in~\eqref{eq:acpf} define an operator $\mathcal{F}_{\mathbf{Y}_{\mathrm{bus}}}$ that maps an input vector $\mathbf{x}$ to an output vector $\mathbf{y}$ \emph{for a fixed network topology}. The topology enters this operator through the bus admittance matrix $\mathbf{Y}_{\mathrm{bus}}$, which is assembled from the connectivity of buses in $\mathcal{N}$ along the edges in $\mathcal{E}$. Both the iterative NRLF solver and the surrogate $f_w$ in~\eqref{eq:pf_functional} are inherently tied to a particular topology: NRLF requires $\mathbf{Y}_{\mathrm{bus}}$ to assemble its Jacobian at every iteration, and the surrogate is trained on a dataset $\{(\mathbf{x}^{(i)}, \mathbf{y}^{(i)})\}_{i=1}^{N}$ whose targets $\mathbf{y}^{(i)}$ are all solutions of~\eqref{eq:acpf} with a single $\mathbf{Y}_{\mathrm{bus}}$. The parameters $w$ are therefore fitted to the input--output mapping induced by that one topology, and the surrogate performs accurately on it under load uncertainty. Throughout, $\mathbf{y}$ denotes the true ACPF solution, whereas $\widehat{\mathbf{y}}=f_w(\mathbf{x})$ denotes the surrogate prediction.

When the operating topology changes, this assumption is no longer valid. A line outage corresponds to an $N$-$k$ contingency, where $k$ denotes the number of lines removed from the network. Removing any edge from $\mathcal{E}$ alters $\mathbf{Y}_{\mathrm{bus}}$, which reroutes power through the remaining lines and produces different voltage magnitudes and angles throughout the network even for the very same injection profile $\mathbf{x}$ (depending on the pre-outage line flows). Accordingly, the physical operator $\mathcal{F}_{\mathbf{Y}_{\mathrm{bus}}}$ is replaced by a new operator $\mathcal{F}^{(N\text{-}k)}$, and the dataset that NRLF would generate on the contingency topology is therefore \emph{not} the dataset $f_w$ was trained on. We refer to this change in the underlying physical mapping as an \emph{operator shift}.

Importantly, operator shift is not merely a change in the input distribution. Let $\mathbf{x}\sim P_X$ denote operating conditions and define $\mathbf{y}=\mathcal{F}_{\mathbf{Y}_{\mathrm{bus}}}(\mathbf{x})$. Then each topology induces an output distribution over solutions $\mathbf{y}$. Even when $P_X$ is identical across topologies, the induced distribution of $\mathbf{y}$ changes because the underlying operator depends on $\mathbf{Y}_{\mathrm{bus}}$, which is directly altered by the outage. Denoting the resulting output distributions under the base and contingency topologies as $\mathcal{D}_{\mathrm{base}}$ and $\mathcal{D}_{N\text{-}k}$, respectively,
\begin{equation}
  \mathcal{D}_{\mathrm{base}}
  \;\neq\;
  \mathcal{D}_{N\text{-}k}.
  \label{eq:dist_shift}
\end{equation}

This mismatch in output distributions directly translates into prediction failure: a surrogate $f_w$ whose weights were fitted to $\mathcal{D}_{\mathrm{base}}$ is systematically inaccurate when evaluated on samples drawn from $\mathcal{D}_{N\text{-}k}$, because the mapping it has internalized is no longer the one generating the test data. One way to address this is to train the model on all possible topologies. However, the number of $N$-$k$ contingencies grows combinatorially with network size, making this computationally prohibitive. 
We instead seek a way to match these distributions at test time without retraining and arrive naturally at the central question of this paper: \textit{what statistical information from the test-time topology is sufficient to account for the operator shift induced by topology changes?}

\begin{remark}[Graph-based surrogates and operator shift]
\label{rem:gnn_shift}
GNN surrogates such as PFGNN~\cite{lin2024} differ from the vectorized surrogate \eqref{eq:pf_functional} as they accept the contingency-specific adjacency, or equivalently the post-outage $\mathbf{Y}_{\mathrm{bus}}$, as an input at inference time. Their forward pass is therefore nominally aware of the topology change. This alone, however, does not resolve the operator shift. The message-passing weights of the GNN are still fitted to the joint distribution of $(\mathbf{x}, \mathbf{Y}_{\mathrm{bus}}, \mathbf{y})$ realised during training; when the test-time $\mathbf{Y}_{\mathrm{bus}}$ corresponds to a topology that was \emph{not} represented in the training set, the network is asked to extrapolate the learned mapping to an unseen $\mathbf{Y}_{\mathrm{bus}}$, and supplying the new admittance matrix as input is not in itself sufficient to recover the post-contingency operator $\mathcal{F}^{(N\text{-}k)}$. Section~\ref{sec:sota-comparison} confirms this empirically: even with the contingency $\mathbf{Y}_{\mathrm{bus}}$ available at inference, PFGNN exhibits the same one-to-two-orders-of-magnitude degradation under $N$-$1$ and $N$-$2$ outages as the topology-agnostic baselines. Operator shift therefore needs to be corrected at the level of the output distribution, independently of whether the backbone consumes the contingency graph.
\end{remark}

\section{Proposed Method: In-Context Whitening}
\label{sec:icw}
A surrogate trained on the base topology learns to predict the ACPF solution distribution induced by that topology. Under an $N$-$k$ contingency, the physical mapping changes and the resulting outputs occupy a different region of the solution space (Section~\ref{subsec:operator_shift}). Our goal is therefore to \emph{re-anchor} contingency-topology outputs to a topology-invariant whitened space using a lightweight, test-time procedure that requires no gradient updates, adds negligible per-sample cost, and preserves a closed-form inverse.

A natural way to reduce cross-topology mismatch is to align low-order statistics of the output distribution. In particular, shifting and linearly re-scaling/rotating outputs can exactly align the first two moments (mean and covariance), which typically capture the dominant operating-point and scaling/correlation differences induced by topology changes.

\subsection{Restriction to Affine Transformations}
\label{subsec:affine}
We seek a transformation that makes outputs from different topologies comparable while keeping inference fast and exactly invertible. The constraint that decides what kind of map we can afford is computational, not statistical: the benefit of replacing NRLF with $f_w$ evaporates if the pre- and post-processing around $f_w$ is heavier than the forward pass itself or lacks a closed-form inverse. General moment-matching maps beyond the first two moments are nonlinear; such nonlinear corrections typically forfeit a closed-form inverse, making test-time inversion as expensive as an iterative solve.

This motivates restricting to affine transformations. Affine maps have closed-form inverses, carry exactly enough parameters to specify a mean vector and a covariance matrix, and are a standard choice for moment-based correction in domain adaptation~\cite{sun2016coral, li2018adabn}. Thus, we consider
\begin{equation}
  \mathbf{z} = T_R(\mathbf{y}) \;=\; \mathbf{W}_R\,(\mathbf{y} - \mathbf{c}_R), \qquad R \in \{\mathcal{D}_{\mathrm{base}}, \mathcal{D}_{N\text{-}k}\},
  \label{eq:affine}
\end{equation}
parameterized by a shift $\mathbf{c}_R \in \mathbb{R}^d$ and a matrix $\mathbf{W}_R \in \mathbb{R}^{d \times d}$. Its inverse $T_R^{-1}(\mathbf{z}) = \mathbf{W}_R^{-1}\mathbf{z} + \mathbf{c}_R$ is a single matrix-vector product, evaluated in constant time per sample with a one-time computation of the matrix inverse.

Within this class of transformations, the first two moments can be matched exactly. Setting $\mathbf{c}_R = \boldsymbol{\mu}_R$ centers $T_R(\mathbf{y})$ at zero mean, and choosing $\mathbf{W}_R$ to satisfy the whitening constraint
\begin{equation}
  \mathbf{W}_R\,\boldsymbol{\Sigma}_R\,\mathbf{W}_R^{\top} \;=\; \mathbf{I}_d
  \label{eq:white_constraint}
\end{equation}
forces its covariance to the identity. Applied separately to $\mathcal{D}_{\mathrm{base}}$ and $\mathcal{D}_{N\text{-}k}$, the constraint leaves their transformed versions with identical first and second moments by construction; any remaining cross-topology discrepancy is therefore confined to third and higher moments.

\subsection{Selecting the Optimal Whitener}
\label{subsec:optimal_whitener}
The whitening constraint~\eqref{eq:white_constraint} does not pin down $\mathbf{W}_R$ uniquely. If any matrix $\mathbf{W}_0$ satisfies \eqref{eq:white_constraint} and $\mathbf{M}$ is any orthogonal matrix ($\mathbf{M}\mathbf{M}^{\top} = \mathbf{I}_d$), then $\mathbf{W} := \mathbf{M}\mathbf{W}_0$ also satisfies it. The feasible set is therefore $\mathcal{W}_R = \{\mathbf{M}\,\boldsymbol{\Sigma}_R^{-1/2}\}$, i.e., any orthogonal rotation of $\boldsymbol{\Sigma}_R^{-1/2}$ is also a valid whitener. Importantly, the choice of $\mathbf{M}$ matters when the output vector contains quantities at different scales. In the power-flow setting the output vector $\mathbf{y} = [\mathbf{p}_g^{\mathrm{s}},\,\mathbf{q}_g,\,\mathbf{v}_l,\,\boldsymbol{\theta}^{\mathrm{ns}}]$ contains physically distinct quantities, each in its own units and with its own downstream interpretation. We want the transformation to best preserve the coordinate-wise semantics of $\mathbf{y}$, so that prediction errors in the whitened space translate directly to errors in the correct physical quantity. Thus, among all whitening candidates, we select the one that minimizes the \textit{mean-squared deviation} between the whitened and the original centered output:
\begin{equation}
  \mathbf{W}_R^{*}
  \;=\;
  \argmin_{\mathbf{W} \in \mathcal{W}_R}\;
  \mathbb{E}_R\bigl\|
    \mathbf{W}(\mathbf{y} - \boldsymbol{\mu}_R)
    - (\mathbf{y} - \boldsymbol{\mu}_R)
  \bigr\|^2.
  \label{eq:WR_opt}
\end{equation}
Below we show that the optimization in~\eqref{eq:WR_opt} admits a unique closed-form solution, and it is exactly known ZCA whitening matrix \cite{kessy2018whitening}.
\begin{proposition}[Optimal whitener]
\label{prop:zca}
The unique solution of~\eqref{eq:WR_opt} is the ZCA whitening matrix
\begin{equation}
  \mathbf{W}_R^{*}
  \;=\;
  \boldsymbol{\Sigma}_R^{-1/2}
  \;=\;
  \mathbf{U}\,\mathrm{diag}(s_1^{-1/2},\ldots,s_d^{-1/2})\,\mathbf{U}^{\top},
  \label{eq:opt_W}
\end{equation}
where $\boldsymbol{\Sigma}_R = \mathbf{U}\,\mathrm{diag}(s_1,\ldots,s_d)\,\mathbf{U}^{\top}$ is the spectral decomposition of $\boldsymbol{\Sigma}_R$.
\end{proposition}

\begin{proof}
Assume $\boldsymbol{\Sigma}_R$ is strictly positive definite, so that $\boldsymbol{\Sigma}_R^{-1/2}$ exists; in the empirical setting of Section~\ref{subsec:icw_procedure} this is enforced by the regularization $\boldsymbol{\Sigma}_{\mathrm{ctx}} + \varepsilon \mathbf{I}$. Let $\mathbf{y}_c := \mathbf{y} - \boldsymbol{\mu}_R$ and $\boldsymbol{\Sigma} := \boldsymbol{\Sigma}_R$, so that $\mathbb{E}[\mathbf{y}_c] = \mathbf{0}$ and $\mathbb{E}[\mathbf{y}_c \mathbf{y}_c^{\top}] = \boldsymbol{\Sigma}$.

\noindent\textbf{Step 1: Simplify the objective.} By the trace identity $\mathbb{E}[\mathbf{y}_c^{\top} \mathbf{A}\,\mathbf{y}_c] = \mathrm{tr}(\mathbf{A}\boldsymbol{\Sigma})$ for any matrix $\mathbf{A}$, expanding the squared norm and taking expectations gives
\begin{equation*}
  \mathbb{E}\|\mathbf{W}\mathbf{y}_c - \mathbf{y}_c\|^2
  \;=\;
  \underbrace{\mathrm{tr}(\mathbf{W}\boldsymbol{\Sigma}\mathbf{W}^{\top})}_{=\,d
    \text{ on } \mathcal{W}_R}
  - 2\,\mathrm{tr}(\mathbf{W}\boldsymbol{\Sigma})
  + \mathrm{tr}(\boldsymbol{\Sigma}).
\end{equation*}
Since $d$ and $\mathrm{tr}(\boldsymbol{\Sigma})$ are constant on $\mathcal{W}_R$, minimizing~\eqref{eq:WR_opt} reduces to maximizing $\mathrm{tr}(\mathbf{W}\boldsymbol{\Sigma})$.

\noindent\textbf{Step 2: Reparametrise in terms of a rotation.} Every $\mathbf{W} \in \mathcal{W}_R$ can be written $\mathbf{W} = \mathbf{M}\boldsymbol{\Sigma}^{-1/2}$ for some orthogonal $\mathbf{M}$, giving $\mathrm{tr}(\mathbf{W}\boldsymbol{\Sigma}) = \mathrm{tr}(\mathbf{M}\boldsymbol{\Sigma}^{1/2})$. The problem reduces to finding the orthogonal matrix $\mathbf{M}$ that maximises $\mathrm{tr}(\mathbf{M}\boldsymbol{\Sigma}^{1/2})$.

\noindent\textbf{Step 3: Show the identity rotation is optimal.} With $\boldsymbol{\Sigma} = \mathbf{U}\mathbf{S}\mathbf{U}^{\top}$ and $\mathbf{S} = \mathrm{diag}(s_1,\ldots,s_d)$, define $\mathbf{A} := \mathbf{U}^{\top}\mathbf{M}\mathbf{U}$, which is orthogonal. By the cyclic property of the trace,
\begin{equation*}
  \mathrm{tr}(\mathbf{M}\boldsymbol{\Sigma}^{1/2})
  \;=\;
  \mathrm{tr}(\mathbf{A}\,\mathbf{S}^{1/2})
  \;=\;
  \sum_{i=1}^{d} A_{ii}\,\sqrt{s_i}.
\end{equation*}
Since $\mathbf{A}$ is orthogonal, each column has unit norm, so $|A_{ii}| \leq 1$, giving $\sum_i A_{ii}\sqrt{s_i} \leq \sum_i |A_{ii}|\sqrt{s_i} \leq \sum_i \sqrt{s_i}$, with equality if and only if $A_{ii} = 1$ for every $i$, hence $\mathbf{M}^* = \mathbf{I}_d$.

\noindent\textbf{Step 4: Conclude.} The unique minimizer of~\eqref{eq:WR_opt} is therefore $\mathbf{W}_R^* = \boldsymbol{\Sigma}_R^{-1/2}$, the ZCA whitening matrix~\cite{kessy2018whitening}.
\end{proof}

The result has a clean geometric reading. The decomposition $\mathbf{W}_R^{*} = \mathbf{U}\,\mathrm{diag}(s_i^{-1/2})\,\mathbf{U}^{\top}$ executes three steps in sequence: rotate to the principal axes of $\boldsymbol{\Sigma}_R$ via $\mathbf{U}^{\top}$, scale each axis by $1/\sqrt{s_i}$ to whiten the covariance, then rotate back to the original physical coordinates via $\mathbf{U}$. The final rotation back is what distinguishes ZCA from PCA whitening, and it is precisely what keeps each whitened coordinate $z_i$ aligned with its physical counterpart $y_i$. For any distribution $R$, we therefore have an explicit, uniquely optimal whitener: it satisfies the whitening constraint~\eqref{eq:white_constraint}, and among all valid choices it minimizes the deviation between the whitened and the original centered coordinates. 


\subsection{In-Context Whitening: Training and Inference}
\label{subsec:icw_procedure}
The optimal whitener $\mathbf{W}_R^* = \boldsymbol{\Sigma}_R^{-1/2}$ requires the population mean $\boldsymbol{\mu}_R$ and covariance $\boldsymbol{\Sigma}_R$ of the output distribution under the test-time topology, neither of which is known in advance for an unseen contingency. We propose to extract them \emph{from context}. A small set of solved ACPF instances drawn from the test-time topology supplies empirical estimates for the population moments, \textit{with no gradient update or change to the surrogate's weights}. Let $\mathcal{D}_R = \{(\mathbf{x}_i, \mathbf{y}_i)\}_{i=1}^{N_R}$ be a set of ACPF solved cases from a topology indexed by $R$, with empirical mean $\bar{\mathbf{y}}_R$ and empirical covariance $\hat{\boldsymbol{\Sigma}}_R$. The corresponding empirical whitening transformation is
\begin{equation}
  \mathbf{z} \;=\; \mathbf{W}_R\,(\mathbf{y} - \bar{\mathbf{y}}_R),
  \qquad
  \mathbf{W}_R \;=\; (\hat{\boldsymbol{\Sigma}}_R + \varepsilon \mathbf{I})^{-1/2},
  \label{eq:icw_whitening}
\end{equation}
where $\varepsilon > 0$ keeps $\mathbf{W}_R$ invertible and guarantees the positive-definiteness needed for $\mathbf{W}_R^* = \boldsymbol{\Sigma}_R^{-1/2}$ to exist (Proposition~\ref{prop:zca}), even when $\hat{\boldsymbol{\Sigma}}_R$ is rank-deficient on small samples. By construction, $\mathbf{z}$ has zero empirical mean on $\mathcal{D}_R$, and empirical covariance $\mathbf{I}_d$ up to the small contraction introduced by the $\varepsilon \mathbf{I}$ regularizer.

\textbf{Training.}\; At the training stage, the ICW is applied once, on the base topology only. The training set $\mathcal{D}_{\mathrm{base}}$ yields $\bar{\mathbf{y}}_{\mathrm{base}}$ and $\mathbf{W}_{\mathrm{base}} = (\hat{\boldsymbol{\Sigma}}_{\mathrm{base}} + \varepsilon \mathbf{I})^{-1/2}$, both computed up front and held fixed throughout training. All training targets are whitened, $\mathbf{z}_{\mathrm{base}}^{(i)} = \mathbf{W}_{\mathrm{base}}(\mathbf{y}^{(i)} - \bar{\mathbf{y}}_{\mathrm{base}})$, and the surrogate learns the mapping $f_w: \mathbf{x} \mapsto \mathbf{z}$ in whitened space rather than $\mathbf{x} \mapsto \mathbf{y}$ in physical space. Because the first two moments of the base distribution are removed before $f_w$ ever sees the ACPF output data, its weights encode only the topology-invariant residual of the ACPF operator.

\textbf{Inference.}\; At test time on a new topology, the same process is applied to a small context set $\mathcal{D}_{\mathrm{ctx}} = \{(\mathbf{x}_i, \mathbf{y}_i)\}_{i=1}^{N_{\mathrm{ctx}}}$ of ACPF-solved cases drawn from that topology, yielding contextual statistics $\bar{\mathbf{y}}_{\mathrm{ctx}}$ and $\mathbf{W}_{\mathrm{ctx}} = (\hat{\boldsymbol{\Sigma}}_{\mathrm{ctx}} + \varepsilon \mathbf{I})^{-1/2}$. For a query input $\mathbf{x}_{\mathrm{test}}$, the frozen model predicts $\widehat{\mathbf{z}}_{\mathrm{test}} = f_w(\mathbf{x}_{\mathrm{test}})$ in the whitened space, and the contextual statistics map it back to the physical scale of the current topology:
\begin{equation}
  \widehat{\mathbf{y}}_{\mathrm{test}}
  \;=\;
  \mathbf{W}_{\mathrm{ctx}}^{-1}\,\widehat{\mathbf{z}}_{\mathrm{test}}
  \;+\;
  \bar{\mathbf{y}}_{\mathrm{ctx}}.
  \label{eq:icw_inversion}
\end{equation}

The model weights $w$ never change. Unlike meta-learners that fine-tune per topology~\cite{chen2022meta}, ICW adds no test-time compute beyond a one-off computation of $\bar{\mathbf{y}}_{\mathrm{ctx}}$, $\hat{\boldsymbol{\Sigma}}_{\mathrm{ctx}}$, and $\mathbf{W}_{\mathrm{ctx}}^{-1}$ from $\mathcal{D}_{\mathrm{ctx}}$. These are computed once per topology and reused across every query on that topology, with no gradient step or architectural change. We call this procedure \emph{In-Context Whitening} (ICW). The proposed ACPF surrogate operates in a topology-invariant whitened space, and the context-dependent affine inversion in~\eqref{eq:icw_inversion} re-anchors its predictions to the physical scale of whichever topology is currently in use. Figure~\ref{fig:icw_diagram} shows the inference flow, and Figure~\ref{fig:operator_shift} shows the whitened-space alignment for the same 30-Bus N-2 contingency setting discussed in Section~\ref{subsec:operator_shift}.


\newcommand{\boxttl}[1]{{\scriptsize\textcolor{blue}{#1}}}
\newcommand{\boxeq}[1]{{\normalsize$#1$}}

\begin{figure}[!t]
\centering
\resizebox{\columnwidth}{!}{%
\begin{tikzpicture}[
    >=Stealth,
    every node/.style={font=\small},
    box/.style={draw, thick, rounded corners=2pt, align=center,
                minimum height=1.15cm},
    data/.style={box, fill=blue!8, draw=blue!35!black},
    basestat/.style={box, fill=blue!6, draw=blue!35!black},
    trainbox/.style={box, fill=violet!14, draw=violet!55!black},
    frozen/.style={box, fill=gray!12},
    ctx/.style={box, fill=orange!16, draw=orange!55!black},
    stat/.style={box, fill=orange!10, draw=orange!55!black},
    outbox/.style={box, fill=green!12, draw=green!45!black},
    arr/.style={->, thick},
    ctxsolid/.style={->, thick, orange!60!black},
    ctxarr/.style={->, thick, orange!60!black, densely dashed},
    deploy/.style={->, thick, violet!55!black, densely dashed},
    lbl/.style={font=\scriptsize, text=black!55},
]

\node[font=\normalsize\bfseries, anchor=west] at (-0.9, 5.35)
    {(a) Training: once, on the base topology};

\node[data, minimum width=2.1cm] (d1) at (0.3, 4.4)
    {\boxttl{Base data}\\[3pt]\boxeq{\{\mathbf{x}_i,\mathbf{y}_i\}_{i=1}^N}};

\node[basestat, minimum width=3.6cm] (s1) at (3.8, 4.4)
    {\boxttl{Fit whitener}\\[3pt]\boxeq{W_{\text{base}}=(\widehat{\Sigma}_{\text{base}}+\varepsilon I)^{-1/2}}};

\node[data, minimum width=3.1cm] (w1) at (7.8, 4.4)
    {\boxttl{Whiten targets}\\[3pt]\boxeq{\mathbf{z}_i=W_{\text{base}}(\mathbf{y}_i-\bar{\mathbf{y}}_{\text{base}})}};

\node[trainbox, minimum width=2.0cm] (t1) at (10.9, 4.4)
    {\boxttl{Train}\\[3pt]\boxeq{f_w:\mathbf{x}\mapsto\mathbf{z}}};

\draw[arr] (d1) -- (s1);
\draw[arr] (s1) -- (w1);
\draw[arr] (w1) -- (t1);

\node[font=\normalsize\bfseries, anchor=west] at (-0.9, 3.2)
    {(b) Inference: per new topology, gradient-free};

\node[data, minimum width=1.7cm] (xtest) at (0,1.0)
    {\boxttl{Query}\\[2pt]\boxeq{\mathbf{x}_{\text{test}}}};

\node[frozen, minimum width=2.5cm] (model) at (3.2, 1.0)
    {\boxttl{Surrogate (frozen)}\\[3pt]\boxeq{f_w}};

\node[ctx, minimum width=3.3cm] (decode) at (7.3, 1.0)
    {\boxttl{Invert}\\[3pt]\boxeq{\widehat{\mathbf{y}}=W_{\text{ctx}}^{-1}\widehat{\mathbf{z}}+\bar{\mathbf{y}}_{\text{ctx}}}};

\node[outbox, minimum width=1.7cm, minimum height=0.95cm] (ytest) at (10.9, 1.0)
    {\boxeq{\widehat{\mathbf{y}}_{\text{test}}}};

\draw[arr] (xtest) -- (model);
\draw[arr] (model) -- (decode) node[midway, above, lbl] {$\widehat{\mathbf{z}}$};
\draw[arr] (decode) -- (ytest);

\draw[deploy] (t1.south) -- ++(0,-1.35) -| (model.north)
    node[pos=0.30, fill=white, text=violet!55!black]
    {\normalsize deploy: $w$ frozen};

\node[ctx, minimum width=2.4cm] (cdata) at (1.0, -1.0)
    {\boxttl{Context}\\[2pt]\boxeq{\{(\mathbf{x}_i,\mathbf{y}_i)\}_{i=1}^{N_{\text{ctx}}}}};

\node[stat, minimum width=2.3cm] (cstat) at (4.3, -1.0)
    {\boxttl{Moments}\\[2pt]\boxeq{\bar{\mathbf{y}}_{\text{ctx}},\,\widehat{\Sigma}_{\text{ctx}}}};

\draw[ctxsolid] (cdata) -- (cstat);
\draw[ctxsolid] (cstat) -|  (decode) node[midway, right,text=black,yshift=0.5cm]
{\boxeq{W_{\text{ctx}}=(\widehat{\Sigma}_{\text{ctx}}+\varepsilon I)^{-1/2}}};

\end{tikzpicture}
}
\caption{In-Context Whitening (ICW). \textbf{(a)} The surrogate is trained once
on the base topology in the whitened output space 
and learn $f_w:\mathbf{x}\to\mathbf{z}$. \textbf{(b)} At inference the frozen $f_w$
predicts $\widehat{\mathbf{z}}$; a small context set from the new topology gives
$W_{\text{ctx}}$, and the inverse transform returns physical units.
Only the transform statistics adapt per topology and the surrogate weights never change and no gradient step is taken.}
\label{fig:icw_diagram}
\end{figure}

\begin{figure}[t]
  \centering
  \includegraphics[width=0.8\linewidth]{ 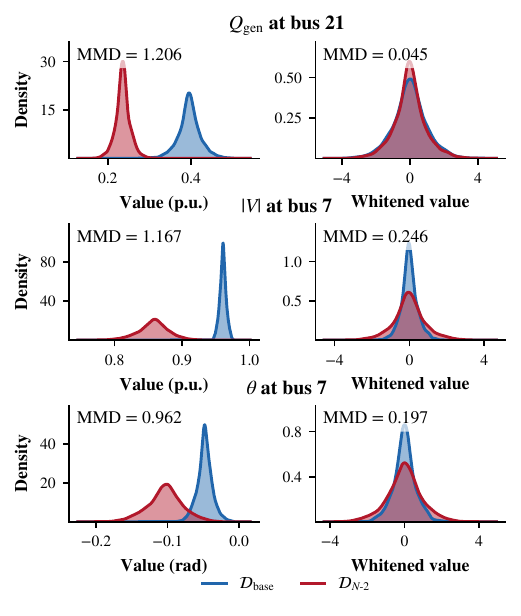}
\vspace{-1em}
\caption{Operator shift on 30-Bus under simultaneous outage of the two highest-loaded lines (N-2 Most contingency), $\pm 50\%$ load, 4k samples per topology. \emph{Left column:} raw outputs $\mathbf{y}$ for the three coordinates with the largest per-dimension MMD; $\mathcal{D}_{\mathrm{base}}$ (blue) and $\mathcal{D}_{\text{N-2}}$ (red) are nearly disjoint. \emph{Right column:} same coordinates after ICW; the distributions collapse, joint MMD drops $28.9\times$.}
\vspace{-1.3em}
\label{fig:operator_shift}
\end{figure}

\begin{remark}[MMD as a diagnostic distance]
Treating each topology as a separate domain places contingency generalization in a standard domain-adaptation setting. One can upper-bound the target-domain loss by a source-domain loss plus a distributional discrepancy term \cite{bendavid2010}:
\begin{equation}
  \mathcal{L}(f_w, \mathcal{D}_{N\text{-}k})
  \;\leq\;
  \mathcal{L}(f_w, \mathcal{D}_{\mathrm{base}})
  \;+\; \mathrm{MMD}_k
  \;+\; \lambda^{\star},
  \label{eq:bendavid}
\end{equation}
where $\mathrm{MMD}_k := \mathrm{MMD}(\mathcal{D}_{\mathrm{base}}, \mathcal{D}_{N\text{-}k})$ is the Maximum Mean Discrepancy (MMD) \cite{gretton2012} and $\lambda^{\star}$ is the minimum combined error achievable by any hypothesis on both domains. In this work, MMD is used only as an \emph{evaluation/diagnostic} measure of alignment quality; the proposed transformation is derived from computational and invertibility constraints together with moment-matching considerations, rather than by optimizing MMD. 
\end{remark}

\begin{remark}[ZCA versus z-score normalization]
\label{rem:zscore}
In \eqref{eq:opt_W}, replacing the full ZCA whitening matrix $\boldsymbol{\Sigma}_R^{-1/2}$ with the diagonal scaling matrix $\mathrm{diag}(\sigma_{R,1}^{-1},\ldots,\sigma_{R,d}^{-1})$ yields per-coordinate z-score normalization, where $\sigma_{R,i}$ denotes the standard deviation of the $i$th coordinate of $\mathbf{y}$ under $R$. This achieves zero mean and unit variance per dimension, but the covariance of the transformed variable is the correlation matrix of $\mathbf{y}$ under $R$, which is identity on the diagonal but has nonzero off-diagonal entries equal to the pairwise correlations. The whitening constraint~\eqref{eq:white_constraint} is therefore only partially satisfied: per-dimension variances are matched, but cross-coordinate covariances are not. Section~\ref{sec:experiments} confirms that z-score normalization yields a smaller MMD reduction than the proposed ZCA whitening.
\end{remark}

The proposed ICW is a nonsingular affine map, and invariants of nonsingular affine transformations, e.g. Mardia's multivariate skewness (a standardized third-order cumulant), are unchanged \cite{mardia1970}. Consequently, after ICW matches the first two moments, any remaining cross-topology mismatch in third and higher moments is not something an additional affine post-processing step can remove. Matrix $\mathbf{W}$ is already determined by the whitening constraint~\eqref{eq:white_constraint} and by the requirement of an explicit inverse used in~\eqref{eq:icw_inversion}. Correcting higher-order discrepancies would therefore require a topology-specific nonlinear map, which would sacrifice closed-form invertibility and substantially increase the number of context samples needed to estimate it reliably. For this reason, and consistent with the diagnostic MMD results and 
ablations in Appendix~\ref{app:two_moments}, we focus on mean--covariance matching as the simplest test-time adaptation that is both computationally cheap and sufficient for the accuracy gains that constitute the paper's main contribution.

\section{Numerical Results and Discussions}
\label{sec:experiments}
This section evaluates ICW under topology changes, surrogate architectures, network sizes, contingency severities, context sizes, and load-perturbation levels. Samples are generated by Newton--Raphson on the corresponding topology. For each load bus $j$, active and reactive demands are perturbed independently as $p_{l,j} = (1 + u_{p,j})\,p_{l,j}^0$ and $q_{l,j} = (1 + u_{q,j})\,q_{l,j}^0$, with $u_{p,j}, u_{q,j} \stackrel{\text{i.i.d.}}{\sim} \mathcal{U}[-\delta, \delta]$; stacking over load buses yields $\mathbf{u}_p,\mathbf{u}_q$ and $\mathbf{p}_l,\mathbf{q}_l$. We report $\delta\in\{0.2,0.5\}$ ($20\%$ or $50\%$), with $\delta=0.2$ as the default. All other inputs in $\mathbf{x}$ are set to the base case. Each test condition uses $1{,}000$ held-out samples, and training uses $4{,}000$--$15{,}000$ base-topology samples depending on network size.

Training mixes two regimes: $50\%$ of samples use fixed $\delta=0.20$, and $50\%$ draw per-sample $\delta \stackrel{\text{i.i.d.}}{\sim} \mathcal{U}[0.05,0.20]$. Context samples match training. Test samples use only the mixed regime ($\delta \sim \mathcal{U}[0.05,0.20]$), reflecting deployment where severity is unknown a priori. Unless stated otherwise, contingency severity is defined by two standardized cases. \emph{N-1 Most} removes the single most heavily loaded line, identified by ranking lines by average loading over $1{,}000$ load scenarios sampled on the base topology; \emph{N-2 Most} removes the two most heavily loaded lines simultaneously. Both are out-of-distribution topologies not seen in training. A line is selected only if its removal keeps the network connected and the resulting topology remains Newton--Raphson convergent at the nominal operating point.

We evaluate three surrogate backbones. An MLP, a Transformer, and PowerFlowNet \cite{lin2024} (a message-passing GNN in per-node prediction space, denoted as PFGNN hereafter). The MLP and Transformer are trained on the flat $d$-dimensional whitened output $\mathbf{z}$; at inference, the frozen model predicts in whitened space and contextual statistics map predictions back to physical scale via~\eqref{eq:icw_inversion}. PFGNN follows \cite{lin2024} without architectural changes; only the data pipeline is modified to incorporate ICW: ZCA whitening is applied to the flat 60-dimensional output before per-node conversion during training, and per-node predictions are converted back to flat space before inverse whitening at inference. ICW integrates with all three backbones without changing model weights or architecture; full specifications and hyperparameters are in Appendix~\ref{app:impl}. To isolate the impact of moment matching, we compare seven preprocessing variants (Table~\ref{tab:preprocessing_configs}): \emph{No Norm}, \emph{Residual}, \emph{Z-Score}, and \emph{ZCA}, each with frozen or context-adaptive statistics when applicable. The proposed ICW method is \emph{ZCA(ctx)}. Performance is measured by mean absolute error (MAE) in per unit on $1{,}000$ held-out test samples. \emph{Overall MAE} is the MAE averaged uniformly across all output quantities $\mathbf{y}$ as defined in Section~\ref{sec:background}.

\begin{table}[t]
\centering
\caption{Evaluated pre-processing configurations }
\label{tab:preprocessing_configs}
\renewcommand{\arraystretch}{1.05}
\begin{tabular}{l|cc}
\hline
\textbf{Moments matched} & \textbf{Frozen stats} & \textbf{Context stats} \\
\hline
None ($\mathbf{W} = \mathbf{I}$, no centring)         & \multicolumn{2}{c}{\emph{No Norm}}            \\
\hline
1st only (mean)                                       & \emph{Residual} & \emph{Residual(ctx)}        \\
1st + per-dim 2nd (variance)                 & \emph{Z-Score Frozen}  & \emph{Z-Score(ctx)}  \\
1st + full 2nd (mean + covariance)               & \emph{ZCA Frozen}      & \emph{\textbf{Proposed ICW}} \\
\hline
\end{tabular}
\end{table}


\begin{table}[h]
\centering
\vspace{-4pt}
\footnotesize
\caption{Overall MAE (p.u.), 30-Bus, 20\% mixed load, 1k test samples. \textbf{Bold} = best per Architecture. \textit{Italics} = context-adaptive (400 ctx samples).}
\label{tab:exp1_arch_invariant}
\small
\begin{tabular}{@{}l@{\hskip -2pt}c@{\hskip 4pt}c@{\hskip 5pt}c@{\hskip 4pt}c@{}}
\toprule
\textbf{Variant} & \textbf{Infer. Time} & \textbf{NR} & \textbf{N2-Most} & \textbf{N1-Most} \\
\cmidrule(lr){3-5}
 & & \textit{Base} & \multicolumn{2}{c}{\textit{Contingency (OOD)}} \\
\midrule
\multicolumn{5}{l}{\textbf{MLP}} \\
No Norm & 0.03s & $2.60 \cdot 10^{-4}$ & $1.552 \times 10^{-2}$ & $5.98 \times 10^{-3}$ \\
Z-Score Frozen & 0.03s & $1.00 \times 10^{-4}$ & $1.554 \times 10^{-2}$ & $5.89 \times 10^{-3}$ \\
Residual & 0.03s & $1.10 \times 10^{-4}$ & $1.554 \times 10^{-2}$ & $5.90 \times 10^{-3}$ \\
ZCA Frozen & 0.03s & $\mathbf{6.00 \times 10^{-5}}$ & $1.554 \times 10^{-2}$ & $5.90 \times 10^{-3}$ \\
\textit{Z-Score(ctx)} & 0.09s & $2.90 \times 10^{-4}$ & $5.90 \times 10^{-4}$ & $4.70 \times 10^{-4}$ \\
\textit{Residual(ctx)} & 0.09s & $2.90 \times 10^{-4}$ & $8.20 \times 10^{-4}$ & $6.30 \times 10^{-4}$ \\
\textit{ICW} & 0.09s & $3.10 \times 10^{-4}$ & $\mathbf{5.60 \times 10^{-4}}$ & $\mathbf{4.60 \times 10^{-4}}$ \\
\midrule
\multicolumn{5}{l}{\textbf{Transformer}} \\
No Norm & 0.04s & $5.80 \times 10^{-4}$ & $1.559 \times 10^{-2}$ & $6.00 \times 10^{-3}$ \\
Z-Score Frozen & 0.04s & $3.10 \times 10^{-4}$ & $1.543 \times 10^{-2}$ & $5.72 \times 10^{-3}$ \\
Residual & 0.04s & $4.40 \times 10^{-4}$ & $1.556 \times 10^{-2}$ & $5.95 \times 10^{-3}$ \\
ZCA Frozen & 0.04s & $\mathbf{1.70 \times 10^{-4}}$ & $1.554 \times 10^{-2}$ & $5.89 \times 10^{-3}$ \\
\textit{Z-Score(ctx)} & 0.13s & $2.90 \times 10^{-4}$ & $6.00 \times 10^{-4}$ & $4.80 \times 10^{-4}$ \\
\textit{Residual(ctx)} & 0.13s & $5.10 \times 10^{-4}$ & $9.00 \times 10^{-4}$ & $7.40 \times 10^{-4}$ \\
\textit{ICW} & 0.13s & $3.30 \times 10^{-4}$ & $\mathbf{5.60 \times 10^{-4}}$ & $\mathbf{4.70 \times 10^{-4}}$ \\
\midrule
\multicolumn{5}{l}{\textbf{PFGNN}} \\
No Norm & 0.11s & $1.38 \times 10^{-3}$ & $2.614 \times 10^{-2}$ & $7.98 \times 10^{-3}$ \\
Z-Score Frozen & 0.11s & $\mathbf{5.00 \times 10^{-4}}$ & $1.569 \times 10^{-2}$ & $6.05 \times 10^{-3}$ \\
Residual & 0.12s & $1.16 \times 10^{-3}$ & $1.569 \times 10^{-2}$ & $6.02 \times 10^{-3}$ \\
ZCA Frozen & 0.11s & $5.40 \times 10^{-4}$ & $1.600 \times 10^{-2}$ & $6.25 \times 10^{-3}$ \\
\textit{Z-Score(ctx)} & 0.29s & $\mathbf{5.00 \times 10^{-4}}$ & $\mathbf{1.85 \times 10^{-3}}$ & $\mathbf{1.36 \times 10^{-3}}$ \\
\textit{Residual(ctx)} & 0.29s & $1.08 \times 10^{-3}$ & $2.29 \times 10^{-3}$ & $1.93 \times 10^{-3}$ \\
\textit{ICW} & 0.29s & $6.10 \times 10^{-4}$ & $2.04 \times 10^{-3}$ & $1.55 \times 10^{-3}$ \\
\bottomrule
\end{tabular}
\end{table}

\begin{figure*}[t]
    \centering
    \includegraphics[width=\linewidth]{ 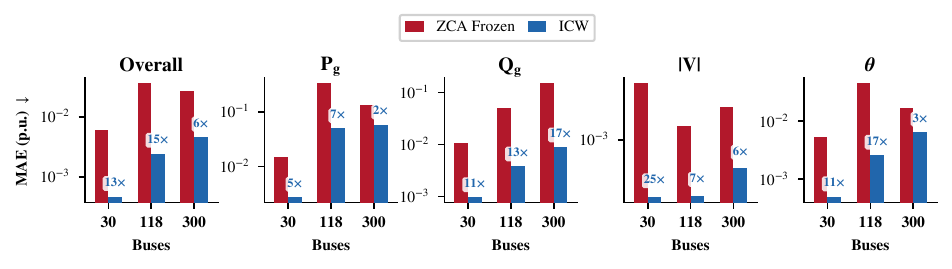}
    \caption{Per-metric MAE (p.u.) for ZCA Frozen vs. ICW, under N-1 Most-Severe contingencies for IEEE bus case 30, 118, and 300 networks. Annotated numbers denote the ZCA/ICW improvement factor.}
  \label{fig:ff_fig4a_n1}
  \end{figure*}

\subsection{Performance Gains by In-Context Whitening}
Table~\ref{tab:exp1_arch_invariant} reports Overall MAE on the 30-bus system under the base, N-1 Most, and N-2 Most conditions across the three backbones. The pattern is a clean split between the frozen and context-adaptive families. Every frozen variant, regardless of normalization, saturates at an N-2 Most floor of roughly $1.5\times10^{-2}$~p.u.: fixing the statistics at training time removes the base topology's moments, not those of the contingency distribution the model is evaluated on. Re-estimating the same statistics from a context set collapses this floor by one to two orders of magnitude on all three backbones, confirming that the shift is carried by the low-order output moments and is correctable at the output level with no weight change. This costs only a modest inference-time increase (MLP: $0.03$\,s to $0.09$\,s), since the statistics are computed once per topology and reused across queries.

Among the adaptive variants, ICW attains the lowest contingency MAE for two of the three backbones at zero gradient cost, reaching $5.60\times10^{-4}$~p.u. (N-2 Most) and $4.60\times10^{-4}$~p.u. (N-1 Most) on the MLP, a $28\times$ and $13\times$ reduction over the ZCA Frozen baseline. The PFGNN is the lone exception, examined in Section~\ref{sec:method-comparison}. Notably, the corrected MLP matches or beats the Transformer and graph-aware PFGNN, indicating the accuracy lost under a topology change is a distributional-alignment problem rather than a capacity one.

Figure~\ref{fig:ff_fig4a_n1} extends this to the 118- and 300-bus systems under N-1 Most at $N_{\mathrm{ctx}}=400$ and $4000$ respectively. ICW reduces Overall MAE by $13\times$, $15\times$, and $6\times$ across the three systems, with gains in $|V|$ by $25\times$ on 30-Bus, $Q_g$ by $17\times$ on 300-Bus. The improvement holds at every system size and across all four physical quantities, so it is the whitening of the full output vector, not any single block, that recovers accuracy.

\subsection{Comparison Against Other Adaptation Strategies}
\label{sec:method-comparison}
Having established that any context-adaptive variant beats its frozen counterpart, we compare ICW against the other adaptive strategies in Table~\ref{tab:exp1_arch_invariant} to isolate what full-covariance matching buys over the cheaper partial corrections. The three differ only in how much of the second moment they match: Residual(ctx) removes the contextual mean, Z-Score(ctx) adds per-dimension variance, and ICW matches the full covariance including cross-coordinate correlations.

Each added moment lowers the error. On the MLP under N-2 Most, Residual(ctx), Z-Score(ctx), and ICW reach $8.20$, $5.90$, and $5.60\times10^{-4}$~p.u. respectively; the same ordering holds on the Transformer. The largest step is from mean-only to variance, with the full-covariance step giving a further but smaller gain, exactly as the two-moment argument predicts: most of the shift lives in the mean and per-dimension variance, and the correlations only ZCA captures account for the remaining portion.

The one reversal is on the PFGNN, where Z-Score(ctx) edges ICW ($1.85\times10^{-3}$ vs. $2.04\times10^{-3}$~p.u. under N-2 Most). This suits its per-node structure: a diagonal estimator matches variance without mixing across nodes, whereas a full-covariance estimator must recover cross-node correlations from only $N_{\mathrm{ctx}}=400$ samples. On the three flat-output backbones, matching the full covariance is at least as good as the diagonal approximation and is the safe default. All three variants share the same gradient-free statistics, so this comparison carries no difference in compute; the larger comparison against gradient-based methods is deferred to Section~\ref{sec:cost-retraining}.

\begin{table*}[t]
\centering
\caption{Baseline MAE comparison on IEEE 30-Bus and IEEE case118 (side-by-side). Context-using methods use a 400-sample context set from the target topology.}
\label{tab:30-118-sidebyside}
\footnotesize
\setlength{\tabcolsep}{2.6pt} 
\renewcommand{\arraystretch}{1.05}
\begin{tabular}{lccc ccc|ccc ccc}
\toprule
& \multicolumn{6}{c}{IEEE 30-Bus} & \multicolumn{6}{c}{IEEE case118} \\
\cmidrule(lr){2-7}\cmidrule(lr){8-13}
& \multicolumn{3}{c}{$|V|$ MAE (p.u.)} & \multicolumn{3}{c}{$\theta$ MAE (deg)}
& \multicolumn{3}{c}{$|V|$ MAE (p.u.)} & \multicolumn{3}{c}{$\theta$ MAE (deg)}  \\
\cmidrule(lr){2-4}\cmidrule(lr){5-7}\cmidrule(lr){8-10}\cmidrule(lr){11-13}
Method & Base & N-1 Most & N-2 Most & Base & N-1 Most & N-2 Most & Base & N-1 Most & N-2 Most & Base & N-1 Most & N-2 Most \\
\midrule
\multicolumn{13}{l}{\textit{Without context samples (zero-shot)}}\\
Modular-GP  \cite{caetano2026modulargp} & 1.1{\scriptsize$\times10^{-4}$} & 5.6{\scriptsize$\times10^{-3}$} & 1.6{\scriptsize$\times10^{-2}$} & 7.7{\scriptsize$\times10^{-2}$} & 5.8{\scriptsize$\times10^{-1}$} & 1.5{\scriptsize$\times10^{0}$}
& 9.2{\scriptsize$\times10^{-5}$} & 1.5{\scriptsize$\times10^{-3}$} & 1.8{\scriptsize$\times10^{-3}$} & 2.6{\scriptsize$\times10^{-1}$} & 2.9{\scriptsize$\times10^{0}$} & 3.7{\scriptsize$\times10^{0}$} \\
DeepOPF-FT \cite{zhou2022deepopfft} & 1.6{\scriptsize$\times10^{-5}$} & 5.3{\scriptsize$\times10^{-3}$} & 1.5{\scriptsize$\times10^{-2}$} & 5.0{\scriptsize$\times10^{-3}$} & 2.7{\scriptsize$\times10^{-1}$} & 4.9{\scriptsize$\times10^{-1}$}
& 6.1{\scriptsize$\times10^{-5}$} & 1.5{\scriptsize$\times10^{-3}$} & 1.7{\scriptsize$\times10^{-3}$} & 3.1{\scriptsize$\times10^{-2}$} & 2.6{\scriptsize$\times10^{0}$} & 3.2{\scriptsize$\times10^{0}$} \\
PFGNN  \cite{lin2024} & 6.4{\scriptsize$\times10^{-5}$} & 5.0{\scriptsize$\times10^{-3}$} & 1.5{\scriptsize$\times10^{-2}$} & 2.5{\scriptsize$\times10^{-2}$} & 3.0{\scriptsize$\times10^{-1}$} & 3.2{\scriptsize$\times10^{-1}$}
& 7.4{\scriptsize$\times10^{-5}$} & 1.4{\scriptsize$\times10^{-3}$} & 1.6{\scriptsize$\times10^{-3}$} & 2.6{\scriptsize$\times10^{-1}$} & 2.6{\scriptsize$\times10^{0}$} & 3.2{\scriptsize$\times10^{0}$} \\
\midrule
\multicolumn{13}{l}{\textit{With context samples}}\\
MMNP \cite{ly2025} & 3.4{\scriptsize$\times10^{-4}$} & 4.7{\scriptsize$\times10^{-4}$} & 5.9{\scriptsize$\times10^{-4}$} & 1.9{\scriptsize$\times10^{-2}$} & 2.9{\scriptsize$\times10^{-2}$} & 3.2{\scriptsize$\times10^{-2}$}
& 4.0{\scriptsize$\times10^{-4}$} & 4.2{\scriptsize$\times10^{-4}$} & 4.2{\scriptsize$\times10^{-4}$} & 4.1{\scriptsize$\times10^{-1}$} & 4.3{\scriptsize$\times10^{-1}$} & 4.3{\scriptsize$\times10^{-1}$} \\
TTF-Ensemble & 1.6{\scriptsize$\times10^{-4}$} & 1.8{\scriptsize$\times10^{-4}$} & 2.4{\scriptsize$\times10^{-4}$} & 2.4{\scriptsize$\times10^{-2}$} & 2.7{\scriptsize$\times10^{-2}$} & 2.7{\scriptsize$\times10^{-2}$}
& 4.0{\scriptsize$\times10^{-4}$} & 4.1{\scriptsize$\times10^{-4}$} & 4.1{\scriptsize$\times10^{-4}$} & 8.4{\scriptsize$\times10^{-1}$} & 8.5{\scriptsize$\times10^{-1}$} & 8.5{\scriptsize$\times10^{-1}$} \\
Proposed (ICW) & 7.3{\scriptsize$\times10^{-5}$} & 2.0{\scriptsize$\times10^{-4}$} & 3.2{\scriptsize$\times10^{-4}$} & 2.0{\scriptsize$\times10^{-2}$} & 2.8{\scriptsize$\times10^{-2}$} & 3.1{\scriptsize$\times10^{-2}$}
& 2.0{\scriptsize$\times10^{-4}$} & 2.1{\scriptsize$\times10^{-4}$} & 2.1{\scriptsize$\times10^{-4}$} & 1.3{\scriptsize$\times10^{-1}$} & 1.5{\scriptsize$\times10^{-1}$} & 1.5{\scriptsize$\times10^{-1}$} \\
\bottomrule
\end{tabular}
\end{table*}

\subsection{Comparison Against Existing Methods}\label{sec:sota-comparison}

Table~\ref{tab:30-118-sidebyside} places ICW against five recent methods on IEEE 30-Bus and case118, split into zero-shot methods that receive no target-topology data (Modular-GP~\cite{caetano2026modulargp}, DeepOPF-FT~\cite{zhou2022deepopfft}, PFGNN~\cite{lin2024}) and context-using methods given a 400-sample context from the target topology (MMNP~\cite{ly2025}, TTF-Ensemble~\cite{jia2024}, and ICW). The two context baselines adapt differently: MMNP is trained on the base topology and conditions on the context purely at inference, taking no gradient step as ICW does, whereas TTF-Ensemble fits a base MLP and then adapts it to each topology with gradient updates on the same context set. The zero-shot methods reproduce the operator-shift collapse of Section~\ref{subsec:operator_shift}: on 30-Bus, PFGNN's $|V|$ error rises from $6.4\times10^{-5}$ to $1.5\times10^{-2}$~p.u. between base and N-2 Most, and on case118 all three exceed $2.6$~deg of angle error under contingency. Every context-using method instead holds its base-case accuracy across N-1 and N-2, confirming that a target-topology context removes the shift regardless of how that context is consumed.

Among the context methods, ICW is the most accurate on case118, roughly $2\times$ below both MMNP and TTF-Ensemble on $|V|$ ($2.1\times10^{-4}$ vs. $4.2$ and $4.1\times10^{-4}$~p.u. under N-2 Most) and $3$ to $6\times$ below them on $\theta$. On 30-Bus it gives the lowest base $|V|$ error ($7.3\times10^{-5}$~p.u.) and stays within the same band as the others under contingency, with TTF-Ensemble marginally lower on the most severe cases ($2.4\times10^{-4}$ vs. $3.2\times10^{-4}$~p.u. on N-2 Most $|V|$). The point is not unconditional accuracy but where ICW sits in cost: it equals or improves on MMNP, which is also gradient-free, and it matches gradient-adapted TTF-Ensemble while paying only the context-encoding cost that TTF-Ensemble incurs on top of its per-topology gradient updates. No competing method is simultaneously as accurate and as cheap to adapt, the cost-accuracy position quantified in Section~\ref{sec:cost-retraining}.

\begin{figure}[t]
    \centering
    \includegraphics[width=0.9\linewidth]{ 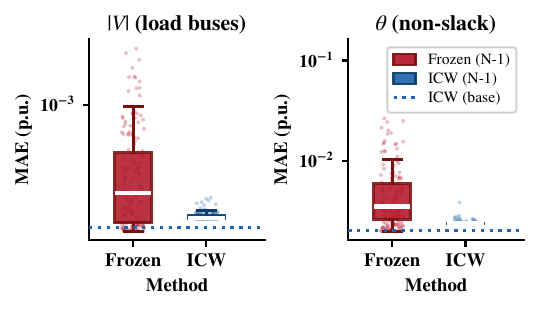}
    \vspace{-1em}
    \caption{Consistency across all 166 N-1 line outages on case118 (20\% load, $N_{\mathrm{ctx}}=400$). Box plots of per-contingency MAE for ZCA Frozen and ICW predictions, shown for $|V|$ at load buses and $\theta$ at non-slack buses; boxes give the median and interquartile range, faint dots individual contingencies, and the dotted line marks ICW's MAE on the base (no-outage) topology as a reference floor. ICW stays near its base-case accuracy across the full set, while the frozen baseline shows wide spread and heavy upper tails under unseen topologies.}
      \label{fig:consistency}
\end{figure}

\subsection{Robustness Across the Full Contingency Set}
\label{sec:full-contingency}
The results so far report two representative contingencies, N-1 Most and N-2 Most. We now check whether the improvement holds across every single-line outage, not only the most severe. Figure \ref{fig:consistency} repeats this at scale on case118 across all 166 N-1 line outages, reported separately for $|V|$ at load buses and $\theta$ at non-slack buses. ICW holds close to its own base-case accuracy over the entire set: its per-contingency distribution stays a tight band near the no-outage value (dotted line) for both quantities. The frozen baseline instead spreads widely and develops a heavy upper tail, with individual contingencies degrading by an order of magnitude or more. The separation is uniform across the full set rather than driven by a few outages, confirming that the operator-shift correction is a property of the whitening transform and not of any particular topology. Physical validity across this same set, that the corrected predictions also satisfy the power-flow equations more closely, is verified in Appendix~\ref{app:physical-validity}. 


\subsection{How Much Context Is Enough}\label{sec:context-size}
The context size $N_{\mathrm{ctx}}$ is the only deployment-time choice ICW requires, and it trades accuracy against the NR-solve cost of collecting the context. Figure~\ref{ff_fig6_c118_stacked} makes this tradeoff explicit on case118, plotting Overall MAE and context-collection wall time on shared axes against $N_{\mathrm{ctx}}$, each point averaged over five independently drawn context sets. MAE falls with more context but flattens quickly, while cost rises steadily, so the two curves together locate where added context stops paying. The lower panel tracks the Log-Euclidean distance $\|\log\boldsymbol{\Sigma}(N)-\log\boldsymbol{\Sigma}(N_{\max})\|_F$ between the context covariance and a converged reference at $N_{\max}=2000$; since the ZCA whitener $\mathbf{W}=\mathbf{U}\bm{\Lambda}^{-1/2}\mathbf{U}^\top$ depends only on this covariance, the distance reaching zero marks the point beyond which additional context carries no new geometric information. At $N_{\mathrm{ctx}}=400$ (dashed line) the covariance distance has already fallen to under half its value at $N=100$ and the MAE curve has nearly flattened, at roughly one-fifth the cost of $N_{\mathrm{ctx}}=2000$. This captures about two-thirds of the achievable MAE improvement on case118; the remaining gain sits in the tail eigenvalues, which its $d=236$ output dimension needs more samples to settle. We therefore use $N_{\mathrm{ctx}}=400$ throughout, a reasonable operating point rather than a convergence requirement, and reserve the eigenvalue-level detail for Appendix~\ref{app:covariance-convergence}.

\begin{figure}[t]
   \centering
   \includegraphics[width=\linewidth]
   { 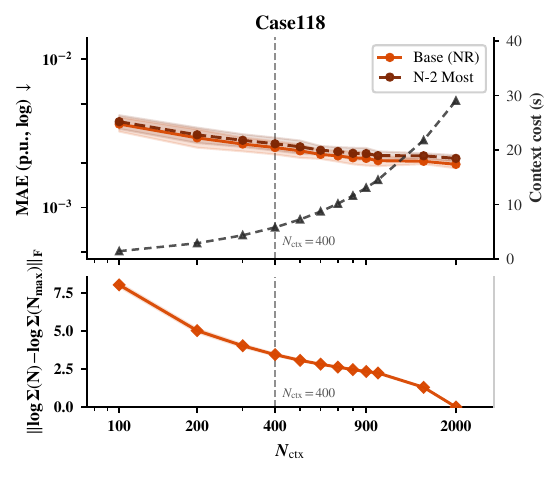}
    \caption{\textbf{Case118: accuracy-cost tradeoff with context size.} \emph{Top:} Overall MAE (log scale) on Base and N-2 Most vs $N_{\mathrm{ctx}}$, with context-collection time on the right axis; 5-draw mean, $\pm1$ std band. \emph{Bottom:} Log-Euclidean covariance distance from the converged reference at $N_{\max}=2000$. Dashed line marks the $N_{\mathrm{ctx}}=400$ operating point.}
    \label{ff_fig6_c118_stacked}
\end{figure}

\subsection{Cost Relative to Retraining}
\label{sec:cost-retraining}

The value of ICW is not raw accuracy but where that accuracy sits in cost. We measure the cost of adapting to a new topology as $T_{\mathrm{data}}+T_{\mathrm{grad}}+T_{\mathrm{inf}}$, where $T_{\mathrm{data}}$ is the NR-solve time to collect the context set, $T_{\mathrm{grad}}$ is the time spent on gradient steps, and $T_{\mathrm{inf}}$ is the forward-pass cost at test time. ICW and gradient-based adaptation differ in a single term: ICW computes $\boldsymbol{\mu}_R$ and $\boldsymbol{\Sigma}_R$ from the context set and takes no gradient step, so $T_{\mathrm{grad}}=0$ and its cost is $T_{\mathrm{data}}+T_{\mathrm{inf}}$, whereas fine-tuning or scratch training adds a $T_{\mathrm{grad}}$ that grows with the number of steps and is re-paid for every new outage. This is the cost structure of Figure~\ref{fig:cost-position}: ICW sits at the bottom-left of the cost-accuracy plane, while gradient-based methods trace a curve descending from the frozen error as $T_{\mathrm{grad}}$ increases. The per-contingency fine-tuning and scratch-training sweeps behind this claim are reported in Appendix~\ref{app:cost-detail}; here we establish the position on two systems and then at deployment scale.

Figure~\ref{fig:300-Bus_cost} makes the single-contingency comparison explicit on 300-Bus under N-1 Most, using a base-trained MLP adapted three ways. The frozen baseline sits at $2.74\times10^{-3}$~p.u.\ Voltage MAE, confirming the same operator shift seen on smaller systems. ICW reduces this to $4.22$, $4.25$, and $3.94\times10^{-4}$~p.u.\ at $N_{\mathrm{ctx}}\in\{400,1000,2000\}$, at costs of $0.18$, $0.44$, and $0.88$~s, with no gradient steps. Fine-tuning at $4000$ steps reaches only $1.08\times10^{-3}$, $7.18\times10^{-4}$, and $5.33\times10^{-4}$~p.u.\ at those context sizes, $2.55\times$, $1.69\times$, and $1.35\times$ worse than ICW while costing $5$ to $21\times$ more. ICW therefore lies strictly below and to the left of every fine-tuning curve: it is both more accurate and cheaper at every context size tested.

\begin{figure}[t]
\centering
\includegraphics[width=\linewidth]{ 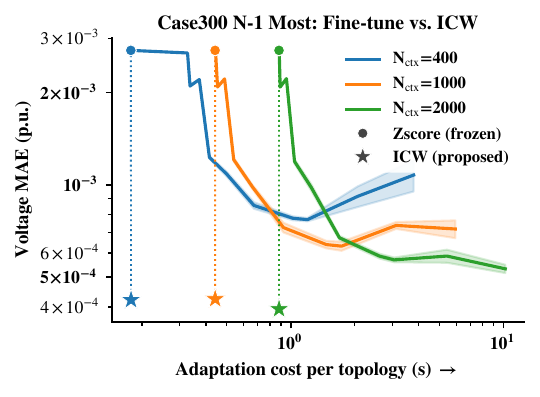}
\caption{Voltage MAE vs.\ per-topology adaptation cost for 300-Bus N-1 Most ($N_\mathrm{ctx}\in\{400,1000,2000\}$, MLP backbone, 3 seeds, $\pm1\sigma$ band). Stars: ICW; dots: frozen baseline at ICW cost; curves: fine-tuning sweep through gradient-step milestones. ICW sits strictly below and to the left of every fine-tuning curve.}
\label{fig:300-Bus_cost}
\end{figure}

This advantage is not confined to easy contingencies. From the $166$ N-1 contingencies on case118 we select $30$ representative ones, $10$ from each of three difficulty terciles ranked by the frozen model's Voltage MAE, and fine-tune the base MLP on each at $N_{\mathrm{ctx}}=400$. For every contingency and step count we form the ratio $r=\mathrm{MAE}_{\mathrm{ICW}}/\mathrm{MAE}_{\mathrm{FT}}$, so $r<1$ means ICW is more accurate than fine-tuning at that cost point. Figure~\ref{fig:case118_band} shows all $30$ curves starting at the frozen parity line $r=1$ and dropping below it as fine-tuning trains: at $4000$ steps the median settles at $r=0.83$, ICW achieving roughly $1.2\times$ lower MAE at $22\times$ lower cost. The margin is nearly identical across the low, mid, and high terciles ($1.20\times$, $1.22\times$, $1.21\times$), so the cost advantage holds uniformly across contingency difficulty.

\begin{figure}[t]
\centering
\includegraphics[width=\linewidth]{ 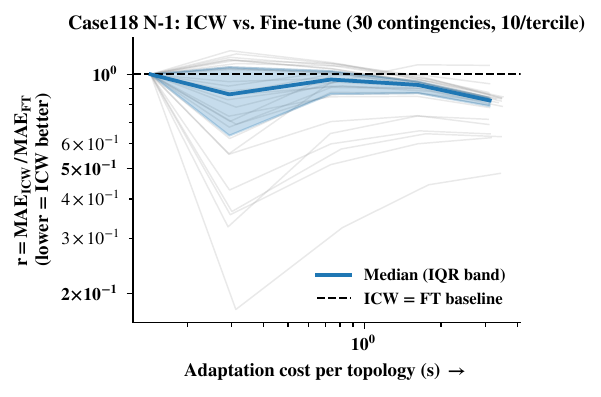}
\caption{\textbf{Case118 N-1: ICW vs.\ fine-tuning across contingency difficulty.} Normalized error ratio $r=\mathrm{MAE}_\mathrm{ICW}/\mathrm{MAE}_\mathrm{FT}$, Overall MAE across $30$ representative N-1 contingencies. Grey lines denote individual contingencies; the blue line and shaded band show the median and IQR. Values below the dashed parity line ($r=1$) indicate lower MAE for ICW. At $4000$ fine-tuning steps the median reaches $r=0.83$ ($1.2\times$ lower MAE), consistently across low-, mid-, and high-difficulty contingencies.}
\label{fig:case118_band}
\end{figure}

The two comparisons above cost a single contingency; in deployment the full N-1 set must be handled. Excluding outages that island the network, this is $K=38$ for 30-Bus, $K=166$ for case118, and $K=253$ for 300-Bus. Since $T_{\mathrm{grad}}=0$ for ICW, its total is $T^{\mathrm{ICW}}_{\mathrm{total}}=K\,(T^{\parallel}_{\mathrm{data}}+T^{\mathrm{ICW}}_{\mathrm{inf}})$, against $T^{\mathrm{FT}}_{\mathrm{total}}=K\,(T^{\parallel}_{\mathrm{data}}+T^{4\mathrm{k}}_{\mathrm{grad}}+T^{\mathrm{FT}}_{\mathrm{inf}})$ for fine-tuning. Table~\ref{tab:sweep_cost_parallel_48} reports both under $48$ parallel CPU cores, where $T^{\parallel}_{\mathrm{data}}$ is the NR-solver time with the context solves distributed across cores, $T^{\mathrm{ICW}}_{\mathrm{inf}}$ and $T^{\mathrm{FT}}_{\mathrm{inf}}$ are the per-contingency inference costs of ICW and fine-tuning respectively, and $T^{4\mathrm{k}}_{\mathrm{grad}}$ is fine-tuning's gradient-step cost at the 4000-step budget used throughout this comparison. ICW completes the full sweep in $4$ to $45$~s across the three systems, against $137$ to $957$~s for fine-tuning at $4000$ steps, a $21$ to $34\times$ speedup. The gap is structural, not just a constant factor: the NR solves that dominate ICW's cost are independent and parallelize across commodity CPU cores, whereas matching that parallelism for fine-tuning would require one GPU per contingency running simultaneously, up to $253$ GPUs for 300-Bus. A single surrogate trained once on the base topology thus adapts to any contingency in the set from context alone, with no retraining, no gradient steps, and an adaptation cost that parallelizes trivially on hardware already present in an operations centre.

\begin{table}[t]
  \centering
  \caption{Total cost of adapting to all N-1 contingencies ($N_\text{ctx}=400$, MLP backbone),
           using \textbf{48 parallel cores} for ACPF context generation.
           $T_\text{data}^{\parallel}$ is the NR solver time with 48 cores running simultaneously
           (serial solver time / 48; process overhead excluded).
           Speedup $= T^\text{FT}_\text{total}/T^\text{ICW}_\text{total}$.}
  \label{tab:sweep_cost_parallel_48}
  \setlength{\tabcolsep}{3.5pt}
  \begin{tabular}{@{}lcrrrrrrr@{}}
    \toprule
    System  & $K$
            & $T_\text{data}^{\parallel}$
            & $T_\text{inf}^\text{ICW}$
            & $T_\text{inf}^\text{FT}$
            & $T_\text{grad}^{4\text{k}}$
            & $T^\text{ICW}_\text{total}$
            & $T^\text{FT}_\text{total}$
            & Speedup \\
            & & (s/ctg) & (s/ctg) & (s/ctg) & (s/ctg) & (s) & (s) & \\
    \midrule
    30-Bus  &  38 & 0.103 & 0.000 & 0.000 & 3.500 &     4 &   137 & $34.3\times$ \\
    118-Bus & 166 & 0.143 & 0.001 & 0.001 & 3.072 &    24 &   534 & $22.3\times$ \\
    300-Bus & 253 & 0.176 & 0.002 & 0.001 & 3.606 &    45 &   957 & $21.3\times$ \\
    \bottomrule
  \end{tabular}
\end{table}

\section{Conclusions}\label{sec:conclusion}
This paper recast contingency adaptation as statistical adaptation: a surrogate trained once on the base topology in an output space free of topology-specific structure, with the contingency-specific statistics reinstated at inference from a small context set. We proposed In-Context Whitening (ICW), an affine transform that matches the first two moments of the output distribution, the most an efficient invertible transform can match, estimating those moments from a few hundred solved NR samples per topology rather than fixing them at training time. Among affine whiteners it is the unique choice preserving the coordinate-wise semantics of the physical output vector (Proposition~\ref{prop:zca}). Across the IEEE 30-, 118-, and 300-bus systems under N-1 and N-2 contingencies, ICW recovers most of the accuracy lost to the operator shift, reducing overall error by 6$\times$ to 28$\times$ over the frozen surrogate and up to 54$\times$ on individual quantities under N-2, and cutting worst-bus power-balance mismatch by up to 30$\times$, with 
consistent gains across three backbones. The residual gap after two-moment matching falls to the finite-sample noise floor, and a lossless third-moment correction yields no further gain, so two moments are empirically sufficient. Against gradient-based adaptation ICW reaches comparable accuracy at only context-collection cost, a $21\times$ to $34\times$ deployment-scale speedup that widens because context collection parallelizes on commodity CPU cores while per-contingency fine-tuning does not. These gains rest on the shift being concentrated in the first two output moments; a discrepancy living in higher moments would fall outside an affine correction, and adaptation still requires NR solves on the new topology, so ICW is gradient-free rather than data-free. Extending output-space whitening to other operator shifts, such as load-regime changes and switching actions, and to larger systems where the covariance tail settles more slowly, is left to future work.

\bibliographystyle{IEEEtran}
\bibliography{main}

\newpage

\newpage
\appendices

\section{Empirical Validation: Two Moments Suffice}
\label{app:two_moments}

This appendix provides the empirical evidence supporting the claim that matching the first two moments is both necessary and sufficient for ICW's accuracy gains, complementing the theoretical argument of Section~\ref{sec:icw}.

\textbf{MMD reduction.} Figure~\ref{fig:two_moments_mmd} tracks the MMD between $\mathcal{D}_{\mathrm{base}}$ and $\mathcal{D}_{N\text{-}2}$ as moments are matched in sequence on the N-2 Most contingency of 30-Bus under $20\%$ load perturbation. The MMD falls from $0.923$ (no matching) to $0.045$ (mean only) to $0.029$ (mean and covariance, i.e.\ ICW). This last value already lies at the sampling noise floor $1/\sqrt{N}\approx0.032$ with 1000 evaluation samples, so the residual cross-topology gap after two-moment matching is statistically indistinguishable from finite-sample noise. By the bound in~\eqref{eq:bendavid}, no further reduction in $\mathrm{MMD}_k$ is achievable with the available samples, and no additional accuracy gain is expected from matching higher moments.

\textbf{MAE saturation.} Figure~\ref{fig:two_moments_mae} confirms this: prediction error drops steeply through level~2 and does not improve at level~3. To rule out the possibility that this saturation is an artifact of an information-losing transform, we apply a lossless sinh--arcsinh (SAS) correction that is a strict bijection verified invertible to machine precision (Appendix~\ref{app:sas}). Even with this lossless correction, MAE rises by $12.7\%$ (MLP) and $33.0\%$ (Transformer) relative to ICW. A capacity sweep confirms the cause: the third-moment target carries identical information to the whitened target, so the fit gap shrinks $34\times$ as model width grows from $128$ to $1024$ units, confirming that the penalty is approximation difficulty introduced by a mild functional nonlinearity in the $\mathbf{x}\to\mathbf{z}$ map, not a recoverable distributional gap. Two moments therefore, suffice: for the N-2 Most contingency, they remove all of the topology-induced shift that any linear, closed-form correction can recover.

\begin{figure}[h]
\centering
\includegraphics[width=\linewidth]{ 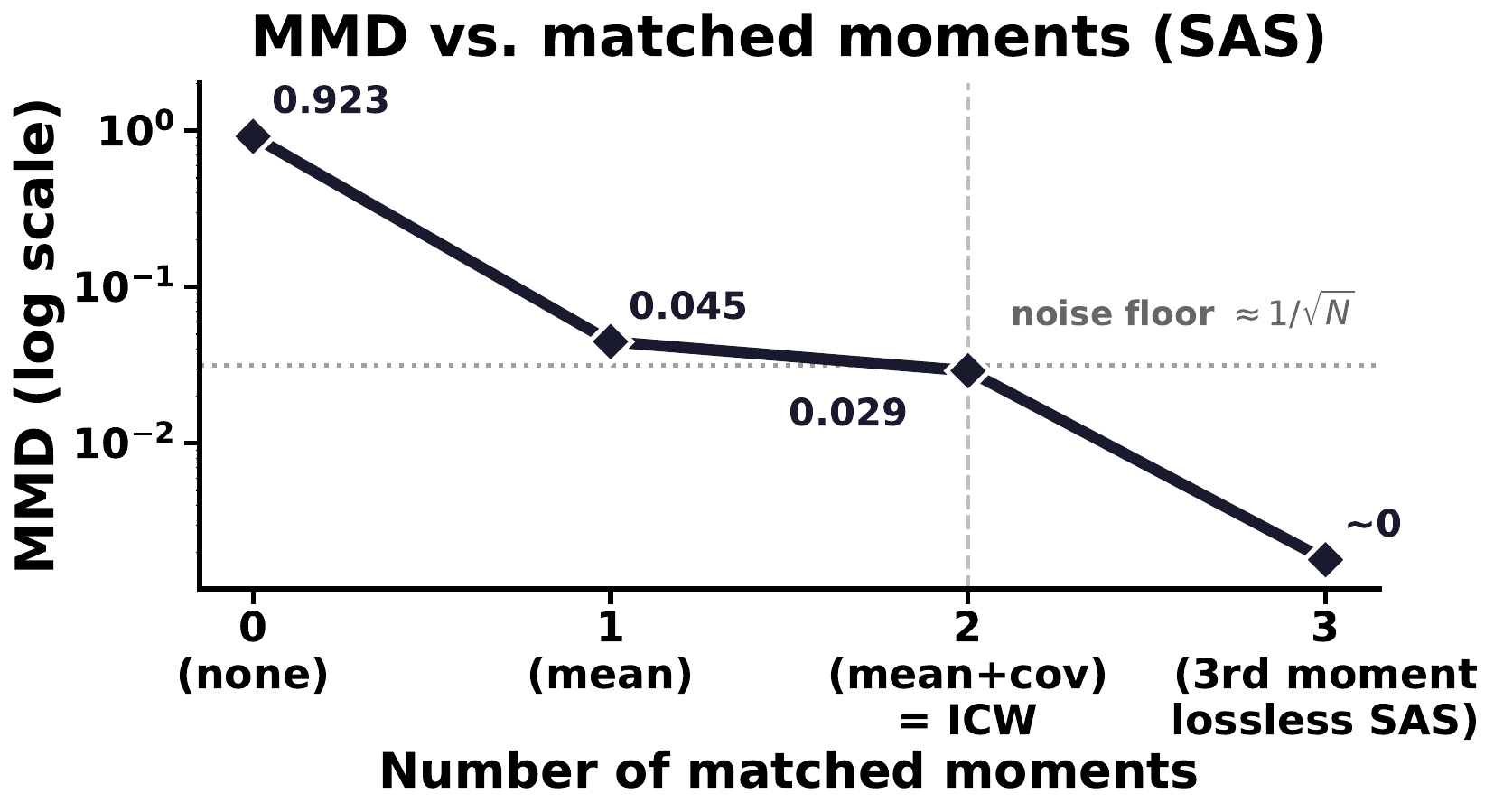}
\caption{MMD between $\mathcal{D}_{\mathrm{base}}$ and
$\mathcal{D}_{N\text{-}2}$ as moments are matched in sequence (30-Bus, 20\%
load, N = 1000 samples). The level~2 value of $0.029$ already lies at the
sampling noise floor $1/\sqrt{N}\approx0.032$, confirming that two moments
saturate the recoverable cross-topology gap.}
\label{fig:two_moments_mmd}
\end{figure}

\begin{figure}[h]
\centering
\includegraphics[width=\linewidth]{ 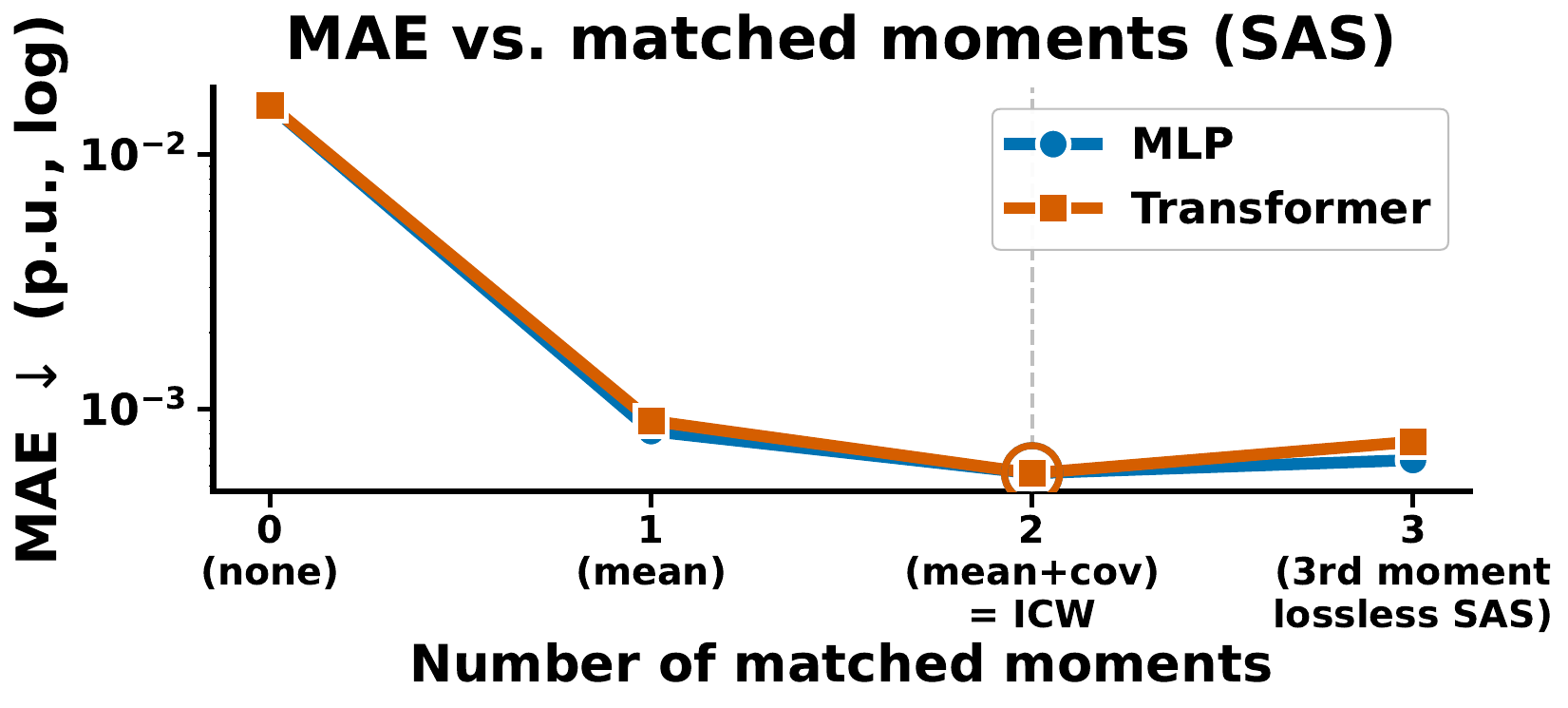}
\caption{MAE versus number of matched moments (30-Bus, 20\% load,
$N_{\mathrm{ctx}}=400$, MLP and Transformer). Error drops steeply through
level~2 and rises at level~3 even with a lossless sinh--arcsinh transform,
confirming that the saturation is not an artifact of information loss.}
\label{fig:two_moments_mae}
\end{figure}

\section{Sinh--Arcsinh Third-Moment Correction}
\label{app:sas}

The standard Cornish--Fisher (CF3) third-moment correction applies the per-dimension map $f(z) = z - \frac{\gamma}{6}(z^2-1)$. Since $f$ is a parabola, it folds at $z^* = 3/\gamma$ and is two-to-one beyond that point. With $\gamma=1$ for example, both $z=2$ and $z=4$ map to $f=1.5$, so the inverse cannot recover which preimage was intended. The sinh--arcsinh (SAS) transform avoids this structural problem by replacing the parabola with a strictly monotone bijection,
\begin{equation}
z_s = \sinh\!\left(\mathrm{arcsinh}(z) - \varepsilon\right),
\quad
z = \sinh\!\left(\mathrm{arcsinh}(z_s) + \varepsilon\right),
\label{eq:sas}
\end{equation}
where $\varepsilon$ is a per-dimension skewness parameter fitted by bisection.

The derivative $dz_s/dz = \cosh\varepsilon - \sinh\varepsilon \cdot z/\sqrt{z^2+1}$ is strictly positive for all $z$ and $\varepsilon$, so the map is a global bijection with no fold and no clipping. Table~\ref{tab:sas_roundtrip} verifies this on real data: the roundtrip error is at machine precision, confirming that no information is lost. Given that the transform is lossless, any remaining fit gap between level~2 and level~3 must come from approximation difficulty rather than missing information. Table~\ref{tab:sas_capacity} confirms this: the $R^2(Z)$ gap between level~2 and level~3 shrinks $34\times$ as model width grows from 128 to 1024 units, the signature of a harder function to approximate, not a harder target to recover. Finally, Figure~\ref{fig:base_vs_n2} shows that the saturation at level~2 holds per output quantity: voltage MAE at level~3 is within $2\%$ of level~2 on both the base and N-2 topologies, confirming that the third-moment correction is neutral across all physical quantities and not just in aggregate.

\begin{table}[h]
\centering
\caption{Round-trip invertibility of the SAS transform on real data (30-Bus,
20\% load, $n=4000$ training samples).}
\label{tab:sas_roundtrip}
\begin{tabular}{lc}
\toprule
\textbf{Check} & \textbf{Value} \\
\midrule
Train roundtrip $|\mathrm{SAS}^{-1}(\mathrm{SAS}(Z)) - Z|$ & $1.88\times10^{-9}$ \\
Inference context roundtrip MAE & $1.97\times10^{-9}$ \\
Max $|\mathrm{SAS}^{-1}(\mathrm{true\ SAS}) - Z|$ (val) & $2.4\times10^{-7}$ \\
\bottomrule
\end{tabular}
\end{table}

\begin{table}[h]
\centering
\caption{Validation $R^2$ and MAE versus model width for level~2 (ZCA) and
level~3 (SAS) targets (30-Bus, 20\% load). The $R^2(Z)$ gap shrinks $34\times$
from width 128 to 1024, confirming the third-moment target carries identical
information but is harder to approximate.}
\label{tab:sas_capacity}
\begin{tabular}{r cc cc c}
\toprule
\textbf{Width} & \textbf{L2 $R^2(Z)$} & \textbf{L3 $R^2(Z)$}
& \textbf{L2 MAE} & \textbf{L3 MAE} & \textbf{$R^2$ gap} \\
\midrule
128  & 0.99190 & 0.98743 & 0.000293 & 0.000477 & 0.00447 \\
256  & 0.99732 & 0.99611 & 0.000179 & 0.000274 & 0.00121 \\
512  & 0.99913 & 0.99874 & 0.000103 & 0.000158 & 0.00039 \\
1024 & 0.99972 & 0.99960 & 0.000059 & 0.000088 & 0.00013 \\
\bottomrule
\end{tabular}
\end{table}

\begin{figure}[t]
\centering
\includegraphics[width=\linewidth]{ 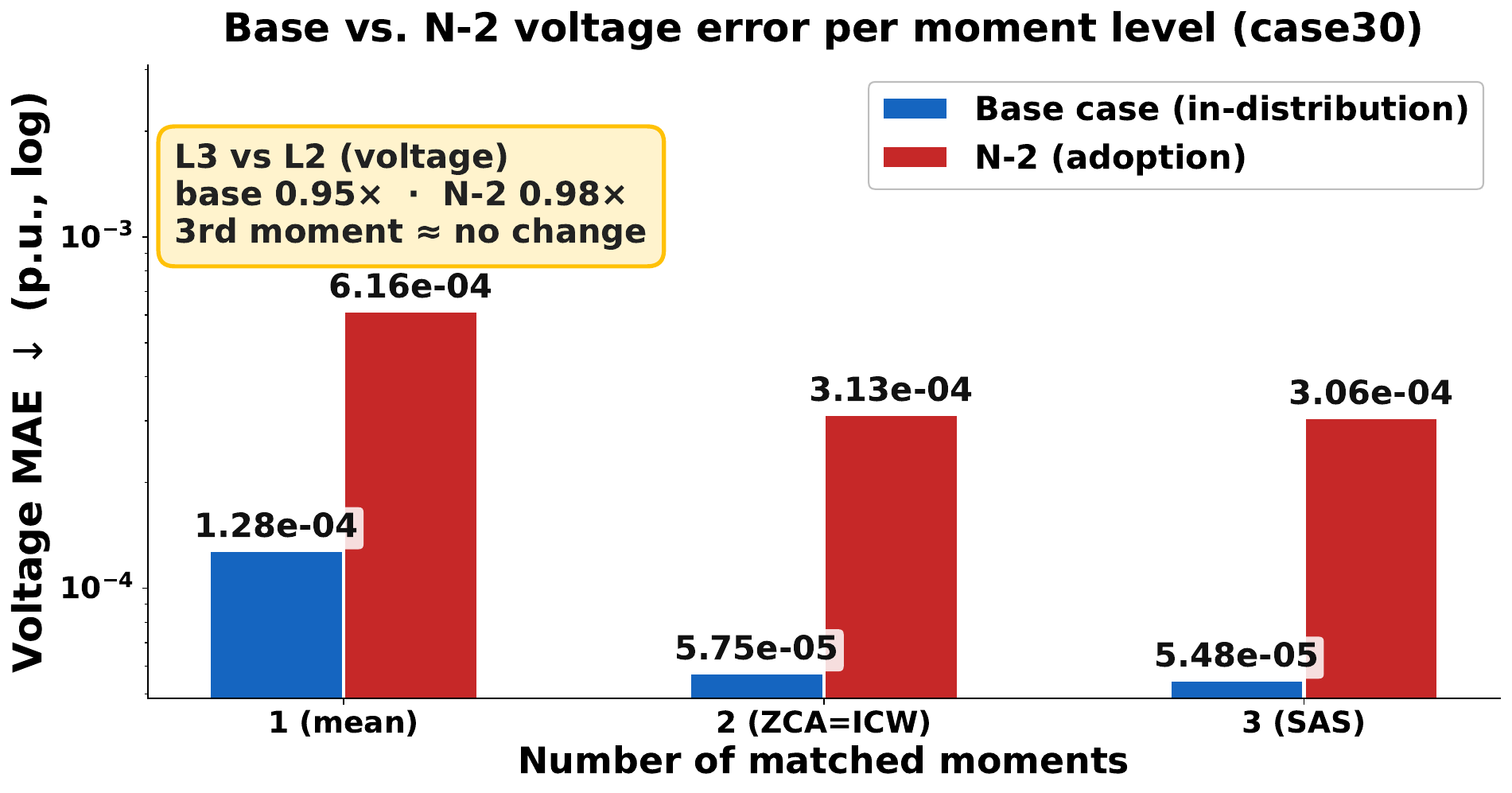}
\caption{Voltage MAE on the base topology and on the N-2 contingency per moment
level (30-Bus, 20\% load). Level~3 tracks level~2 on both domains ($0.95\times$
on base, $0.98\times$ on N-2), confirming that the third-moment correction is
neutral for voltage in and out of distribution.}
\label{fig:base_vs_n2}
\end{figure}

\section{Architecture and Training Hyperparameters}
\label{app:impl}

All models are trained with the AdamW optimiser, CosineAnnealingLR scheduler with $\eta_{\min}=10^{-6}$, batch size 32, gradient clipping at norm 1.0, and weight decay $10^{-5}$. Input $\mathbf{x}$ is always z-score normalised using training statistics, independent of which output preprocessing variant from Table~\ref{tab:preprocessing_configs} is being evaluated. Table~\ref{tab:hyperparams} summarises the architecture-specific hyperparameters.

\begin{table}[h]
\centering
\caption{Architecture specifications and training hyperparameters for all three
backbone models. A dash (--) indicates a parameter that is structurally not
applicable to that architecture (for example, patch size applies only to the
Transformer's patch embedding, and node/edge feature dimensions apply only to
PFGNN's graph structure), not a missing value.}
\label{tab:hyperparams}
\renewcommand{\arraystretch}{1.3}
\setlength{\tabcolsep}{3pt}
\begin{tabular}{l c c c }
\toprule
\textbf{Parameter} & \textbf{MLP} & \textbf{Transf.} & \textbf{PFGNN}~\cite{lin2024}  \\
\midrule
Model / hidden dim. & 256 & 128 & 129  \\
Layers & 2 & 2 & 4 \\
Activation & GELU & GELU & ReLU \\
Dropout & 0.15 & 0.15 & 0.2  \\
Parameters & $\approx$100k & $\approx$292k & $\approx$355k \\
\midrule
Attention heads & -- & 4 & --  \\
Feedforward dim. & -- & 256 & -- \\
Patch size & -- & 12 & --  \\
\midrule
\multicolumn{4}{l}{\textbf{Shared training settings}} \\
\midrule
Optimiser & \multicolumn{3}{l}{AdamW} \\
Scheduler & \multicolumn{3}{l}{CosineAnnealingLR ($\eta_{\min}=10^{-6}$)} \\
Learning rate & \multicolumn{3}{l}{$6.29\times10^{-4}$} \\
Batch size & \multicolumn{3}{l}{32 } \\
Training steps & \multicolumn{3}{l}{40{,}000 (MLP/Transformer/GNN)} \\
Gradient clip & \multicolumn{3}{l}{1.0 (norm)} \\
Weight decay & \multicolumn{3}{l}{$10^{-5}$} \\
\bottomrule
\end{tabular}
\end{table}

\section{Per-Metric Improvement Under N-2 Contingencies}
\label{app:permetric-n2}

Figure~\ref{fig:ff_fig4a_n1} in the main text reports the per-metric frozen-vs-ICW improvement under N-1 Most. Figure~\ref{fig:ff_fig4a_n2} gives the N-2 Most companion across the same three systems. The improvement factors are uniformly larger than in the N-1 setting, consistent with the more severe operator shift a double outage induces: Overall MAE improves by $28\times$, $18\times$, and $8\times$ on 30-, 118-, and 300-Bus, with the largest per-quantity gains on reactive power ($Q_g$ up to $54\times$ on 30-Bus) and voltage magnitude ($|V|$ up to $47\times$). The whitening correction therefore recovers proportionally more accuracy exactly where the frozen surrogate degrades most.

\begin{figure*}[t]
\centering
\includegraphics[width=\linewidth]{ 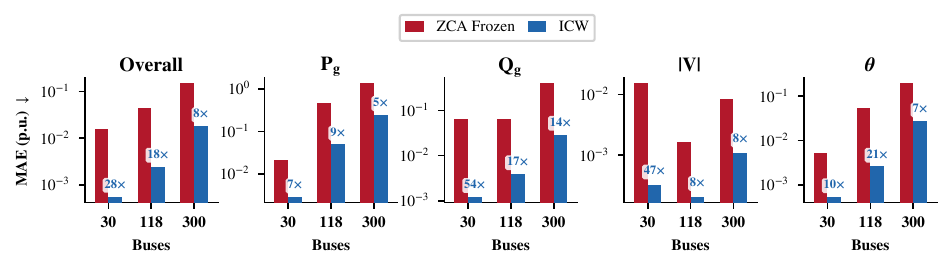}
\caption{Per-metric MAE (p.u.) for Frozen vs.\ ICW under simultaneous N-2
Most-Severe contingencies for IEEE case 30, 118, and 300 networks. ICW delivers
uniformly larger improvements than in the N-1 setting.}
\label{fig:ff_fig4a_n2}
\end{figure*}

\section{N-2 Robustness at the 50\% Load Setting}
\label{app:n2_50pct}

The robustness results in Section~\ref{sec:full-contingency} check every N-1 contingency at the $20\%$ load setting. Here we ask a harder question: does ICW's improvement still hold when both the contingency and the load perturbation are pushed further, to every N-2 double-line-outage pair on 30-Bus at the more severe $50\%$ load setting?

We evaluate the frozen baseline and ICW on all $677$ N-2 contingency pairs that remain connected and Newton--Raphson convergent on 30-Bus at $50\%$ load, the same selection criterion used throughout this paper. To see whether ICW's improvement depends on how hard a given contingency already is for the frozen baseline, we group these $677$ pairs into terciles by frozen-model MAE (low, mid, high severity), the same grouping used for the physical-validity check in Figure~\ref{fig:ff_fig5c}.

Figure~\ref{fig:ff_fig5a_n2} shows the result. ICW reduces MAE in every severity group without exception, with a median improvement factor of $7.6\times$ and a worst-case (bottom $1\%$) factor of $2.3\times$. The improvement is therefore not limited to the easier, lower-severity contingencies, nor to the $20\%$ load setting used elsewhere: it holds even on the hardest N-2 pairs at the more severe $50\%$ perturbation level.

\begin{figure}[t]
\centering
\includegraphics[width=\linewidth]{ 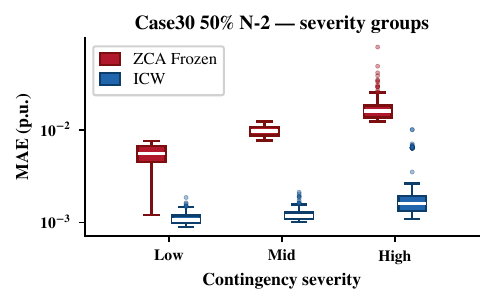}
\caption{MAE distribution of the frozen ZCA baseline and ICW across all $677$
N-2 double-line-outage contingencies on 30-Bus ($50\%$ load perturbation),
grouped into terciles by frozen-model MAE (low, mid, high severity).}
\label{fig:ff_fig5a_n2}
\end{figure}

\section{Distributional Shift and Per-Contingency Error}
\label{app:mmd_gap}

This appendix complements the full-contingency-set robustness results in Section~\ref{sec:full-contingency}. To understand what drives the residual variation in ICW's own accuracy across contingencies, Figures~\ref{fig:ff_fig5b1} and~\ref{fig:ff_fig5b3} plot, for each contingency, the portion of the operator shift ICW closes, $\Delta\mathrm{MMD} = \mathrm{MMD}_Y - \mathrm{MMD}_Z$, where $\mathrm{MMD}_Y$ is the bound's $\mathrm{MMD}_k$ term measured in the raw output space and $\mathrm{MMD}_Z$ is the residual MMD after ZCA whitening, against ICW's own relative error increase over its base-case accuracy, $(e_k - e_{\mathrm{base}})/e_{\mathrm{base}}$, for 30-Bus and case118 respectively. A small number of the hardest contingencies are excluded from each scatter plot only (three for 30-Bus, one for case118), since their relative error increase is computed against a very small base-case ICW MAE denominator, which makes their ratio disproportionately large and would distort the plot's scale; this exclusion affects only these two scatter plots and not the line counts reported elsewhere. Within each system, contingencies where ICW closes less of $\mathrm{MMD}_k$ show a larger relative increase in error and those where it closes more show a smaller increase, exactly the relationship the bound in~\eqref{eq:bendavid} predicts: the more of the distributional shift the transform removes, the closer the contingency-domain error stays to the base-domain error.

\begin{figure}[t]
\centering
\includegraphics[width=\linewidth]{ 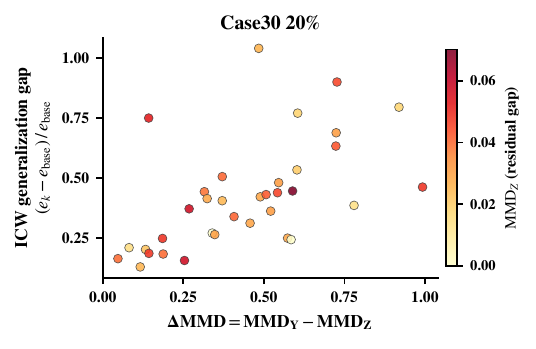}
\caption{Scatter plot of the distributional shift $\Delta\mathrm{MMD}$ vs.\ ICW
generalization gap $(e_k - e_\mathrm{base})/e_\mathrm{base}$ across all 38 N-1
contingencies, 30-Bus (20\% load). Three contingencies were excluded as noted
in the text.}
\label{fig:ff_fig5b1}
\end{figure}

\begin{figure}[t]
\centering
\includegraphics[width=\linewidth]{ 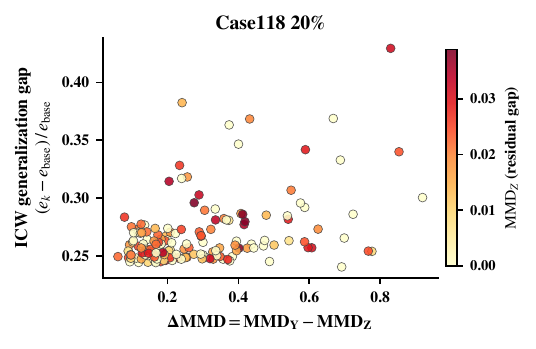}
\caption{Same as Fig.~\ref{fig:ff_fig5b1} for case118 (20\% load, 166 N-1
contingencies). One contingency excluded as noted in the text.}
\label{fig:ff_fig5b3}
\end{figure}

\section{Physical Validity of Predictions}
\label{app:physical-validity}

Section~\ref{sec:full-contingency} establishes that ICW holds its accuracy across the full contingency set in MAE. Here we check physical validity: the worst-bus AC power-balance mismatch $\|f(\mathbf{x}, \widehat{\mathbf{y}})\| = \max_i \sqrt{\Delta P_i^2 + \Delta Q_i^2}$ (excluding the slack bus) for ZCA Frozen and ICW predictions, with contingencies grouped into terciles by frozen-model MAE. As Figure~\ref{fig:ff_fig5c} shows, ICW reduces this mismatch by $14\times$, $18\times$, and $30\times$ for the low-, mid-, and high-severity tiers respectively, with the NR solver baseline itself at approximately $2\times10^{-6}$~p.u., several orders of magnitude below either method. ICW's gains therefore translate into predictions that are closer to satisfying the underlying power-flow equations, not just closer in MAE.

\begin{figure}[t]
\centering
\includegraphics[width=\linewidth]{ 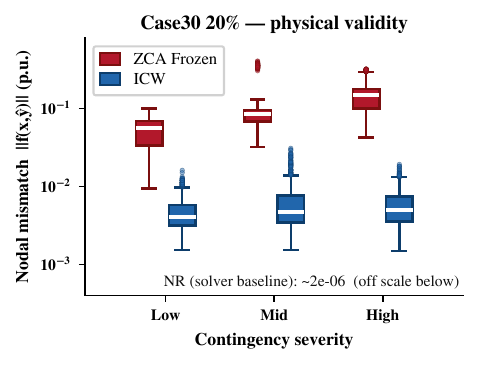}
\caption{Physical validity of predictions. Box plots of the worst-bus AC
power-balance mismatch $|f(x,\hat{y})| = \max_i\sqrt{\Delta P_i^2 + \Delta
Q_i^2}$ (p.u.) for ZCA Frozen and ICW predictions, grouped by contingency severity.
ICW reduces mismatch by $14\times$, $18\times$, and $30\times$ for low-, mid-,
and high-severity groups respectively. The NR solver baseline ($\approx
2\times10^{-6}$~p.u.) lies off scale below.}
\label{fig:ff_fig5c}
\end{figure}

\section{Context Size and Covariance Convergence}
\label{app:covariance-convergence}

This appendix provides the context-size and eigenvalue-level detail supporting the covariance half-convergence argument in Section~\ref{sec:context-size} (How Much Context Is Enough).

Table~\ref{tab:icw_ctx_sweep} reports how ICW's improvement over the frozen baseline scales with context size $N_{\mathrm{ctx}}$ on the 118- and 300-bus systems. On 118-bus, increasing $N_\mathrm{ctx}$ from $400$ to $1600$ raises the Overall improvement factor from $8.3\times$ to $17.2\times$; on 300-bus, increasing $N_\mathrm{ctx}$ from $800$ to $4000$ raises it from $1.4\times$ to $6.0\times$. Since a larger context means more NR solves and higher collection cost, $N_\mathrm{ctx}$ acts as a tunable accuracy-compute tradeoff: a small context already recovers most of ICW's gain, while a larger context yields further improvement at proportionally higher cost. The operator can therefore select $N_\mathrm{ctx}$ according to the desired accuracy-compute tradeoff.

\begin{table}[t]
\centering
\small
\setlength{\tabcolsep}{4pt}
\caption{ICW absolute MAE (p.u.) and improvement factor over frozen baseline
(Frozen\,$\div$\,ICW, $\uparrow$ better), N-1 Most loaded line out contingency.
Bold: largest tested context.}
\label{tab:icw_ctx_sweep}
\begin{tabular}{r rrr rrr}
\toprule
& \multicolumn{3}{c}{\textbf{Absolute MAE (p.u.)}} &
  \multicolumn{3}{c}{\textbf{Improvement ($\times$, $\uparrow$)}} \\
\cmidrule(lr){2-4} \cmidrule(l){5-7}
\multicolumn{1}{c}{{\boldmath$N_{\mathrm{ctx}}$}} &
\multicolumn{1}{c}{\textbf{Overall}} &
\multicolumn{1}{c}{{\boldmath$Q_\mathrm{gen}$}} &
\multicolumn{1}{c}{{\boldmath$|V|$}} &
\multicolumn{1}{c}{\textbf{Overall}} &
\multicolumn{1}{c}{{\boldmath$Q_\mathrm{gen}$}} &
\multicolumn{1}{c}{{\boldmath$|V|$}} \\
\midrule
\multicolumn{7}{l}{\textbf{{118-Bus} \boldmath($d_\mathrm{out} = 236$)}} \\
 400 & 0.00433 & 0.00621 & 0.000221 & \textbf{8.3}  & \textbf{8.2}  & \textbf{6.7} \\
 600 & 0.00415 & 0.00596 & 0.000215 & 8.6  & 8.6  & 6.9 \\
 800 & 0.00413 & 0.00587 & 0.000212 & 8.7  & 8.7  & 7.0 \\
1200 & 0.00269 & 0.00386 & 0.000195 & 13.3 & 13.2 & 7.6 \\
\textbf{1600} & \textbf{0.00208} & \textbf{0.00299} & \textbf{0.000186} &
  \textbf{17.2} & \textbf{17.1} & \textbf{8.0} \\
\midrule
\multicolumn{7}{l}{\textbf{{300-Bus} \boldmath($d_\mathrm{out} = 600$)}} \\
  800 & 0.01995 & 0.02023 & 0.000911 & \textbf{1.4} & \textbf{7.4}  & \textbf{2.8} \\
 1200 & 0.02042 & 0.02003 & 0.000898 & 1.3 & 7.5  & 2.9 \\
 2000 & 0.01925 & 0.01928 & 0.000841 & 1.4 & 7.8  & 3.0 \\
 3000 & 0.00819 & 0.01035 & 0.000512 & 3.3 & 14.5 & 5.0 \\
\textbf{4000} & \textbf{0.00451} & \textbf{0.00886} & \textbf{0.000446} &
  \textbf{6.0} & \textbf{17.0} & \textbf{5.7} \\
\bottomrule
\end{tabular}
\end{table}

Figure~\ref{ff_fig6_c30_stacked} gives the 30-Bus accuracy-cost tradeoff, the
companion to the case118 panel in the main text
(Fig.~\ref{ff_fig6_c118_stacked}). The MAE curve flattens and the Log-Euclidean
covariance distance halves by $N_\mathrm{ctx}=400$, at which point $89\%$ of the
achievable improvement is captured on 30-Bus, against $66\%$ on case118; the
gap is explained by the eigenvalue analysis below.

\begin{figure}[t]
\centering
\includegraphics[width=\linewidth]{ 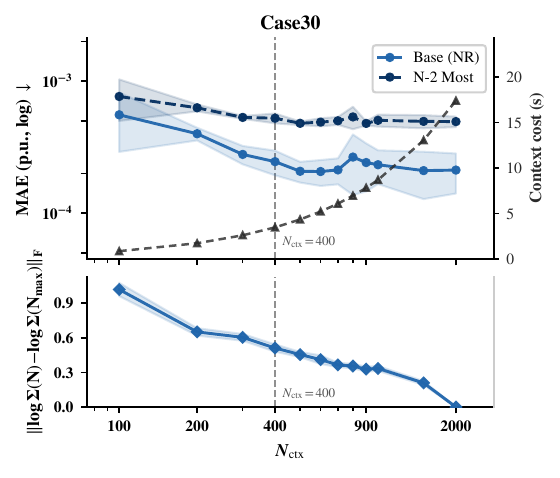}
\caption{\textbf{30-Bus: accuracy-cost tradeoff with increasing context.}
\emph{Top:} Overall MAE vs.\ $N_\mathrm{ctx}$ (log scale) on Base (NR) and N-2
Most contingency, averaged over 5 independently drawn context sets per
$N_\mathrm{ctx}$ (shaded band: $\pm1$ std), with context-collection wall time on
the right axis ($8.69$\,ms/NR-solve). \emph{Bottom:} Log-Euclidean covariance
distance $\|\log\boldsymbol{\Sigma}(N) - \log\boldsymbol{\Sigma}(N_\mathrm{max})\|_F$
from the converged reference at $N_\mathrm{max}=2000$.}
\label{ff_fig6_c30_stacked}
\end{figure}

Figure~\ref{ff_fig6_app_spectrum} shows the ordered eigenvalue spectrum of
$\boldsymbol{\Sigma}(N_\mathrm{ctx})$ for 30-Bus and case118. The spectrum
decays by three to five orders of magnitude across the first $10$--$20$ indices,
giving an effective rank of roughly $12$--$15$ for 30-Bus and $15$--$20$ for
case118, far below their ambient output dimensions of $60$ and $236$. The
leading eigenvalues are visually indistinguishable across all tested
$N_\mathrm{ctx}$, while the tail eigenvalues fan out at small $N_\mathrm{ctx}$
and converge inward as $N$ grows. Since the ZCA whitening matrix
$\mathbf{W} = \mathbf{U}\boldsymbol{\Lambda}^{-1/2}\mathbf{U}^\top$ is dominated
by the leading eigenpairs, the whitening transform stabilises well before the
full spectrum converges, consistent with the large MAE improvement already
achieved at moderate $N_\mathrm{ctx}$ in the main text.

Figure~\ref{ff_fig6_app_deviation} quantifies this at the level of individual
eigenvalues, plotting the relative deviation
$(\lambda_i(N) - \lambda_i(2000))/\lambda_i(2000) \times 100\%$ per index $i$ and
context size $N_\mathrm{ctx}$. Leading indices converge within $\pm5\%$ of the
reference even at $N=100$, confirming that the dominant geometric structure of
$\boldsymbol{\Sigma}$ is captured with minimal context; tail indices show larger
deviation and settle progressively. The deviation of
$\boldsymbol{\Sigma}(1500)$ is near zero across all indices for both systems,
confirming $N_\mathrm{max}=2000$ is itself a valid converged reference rather
than a moving target. The wider oscillations visible for case118 at small
$N_\mathrm{ctx}$, particularly in the mid-to-tail indices, reflect its higher
output dimension and explain why only $66\%$ of the achievable MAE improvement
is captured at $N_\mathrm{ctx}=400$, compared to $89\%$ for 30-Bus.

\begin{figure}[t]
\centering
\includegraphics[width=\linewidth]{ 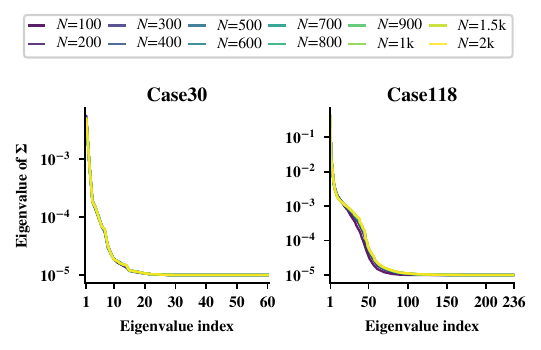}
\caption{\textbf{Eigenvalue spectra of the context covariance
$\boldsymbol{\Sigma}(N_\mathrm{ctx})$.} Ordered spectra for each
$N_\mathrm{ctx}$ (see legend), 30-Bus and case118. Dominant eigenvalues
stabilise by $N_\mathrm{ctx}=200$; tail eigenvalues continue to shift, with
case118 exhibiting a richer tail due to its higher dimensionality ($d=236$).}
\label{ff_fig6_app_spectrum}
\end{figure}

\begin{figure}[t]
\centering
\includegraphics[width=\linewidth]{ 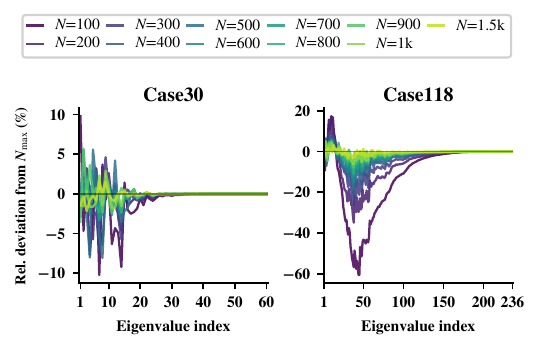}
\caption{\textbf{Per-index eigenvalue deviation from $N_\mathrm{max}=2000$.}
Relative deviation $(\lambda_i(N) - \lambda_i(2000))/\lambda_i(2000) \times
100\%$ per eigenvalue index $i$ and context size $N_\mathrm{ctx}$, 30-Bus and
case118. Leading indices converge within $\pm5\%$ of the reference even at
$N=100$; tail indices settle progressively, with case118 showing broader
oscillations at small $N_\mathrm{ctx}$.}
\label{ff_fig6_app_deviation}
\end{figure}

\section{Cost Detail: Fine-Tuning and Scratch Training}
\label{app:cost-detail}

Section~\ref{sec:cost-retraining} compares ICW against gradient-based adaptation
at deployment scale. This appendix reports the per-single-contingency cost
curves behind that comparison on the 30-Bus N-2 Most contingency, for
fine-tuning from base-trained weights and for training from scratch. We measure
total cost as $T_\mathrm{data} + T_\mathrm{grad} + T_\mathrm{inf}$. For ICW,
$T_\mathrm{grad}=0$ always: at $N_\mathrm{ctx}=400$
($T_\mathrm{data}=4.968$\,s) its total is $5.059$\,s (MLP) and $5.257$\,s
(PFGNN), and at $N_\mathrm{ctx}=1000$ ($T_\mathrm{data}=8.686$\,s) it is
$8.772$\,s and $8.994$\,s. Fine-tuning and scratch training add a
$T_\mathrm{grad}$ that grows with the number of gradient steps.

\textbf{Fine-tuning.} We take the base-trained MLP and PFGNN used by ICW and fine-tune them (no-norm) on the same context set of solved NR samples ICW uses for its inference-time statistics, at $N_\mathrm{ctx}=400$ and $N_\mathrm{ctx}=1$k, tracking Voltage MAE against total wall-clock cost as the number of steps grows. Figure~\ref{ff_figure3_combined} shows the ICW star of both context sizes at the bottom-left of each panel, the lowest MAE at the lowest cost. Fine-tuning's curve approaches but does not cross below ICW's accuracy within the tested step range, even as its cost grows past $30\times$ ICW's. Fine-tuning starts from the same base-trained weights ICW uses, so it begins closer to a usable solution than a randomly initialized model, but it still pays a per-contingency gradient-based cost that ICW does not.

\begin{figure}[t]
\centering
\includegraphics[width=\linewidth]{ 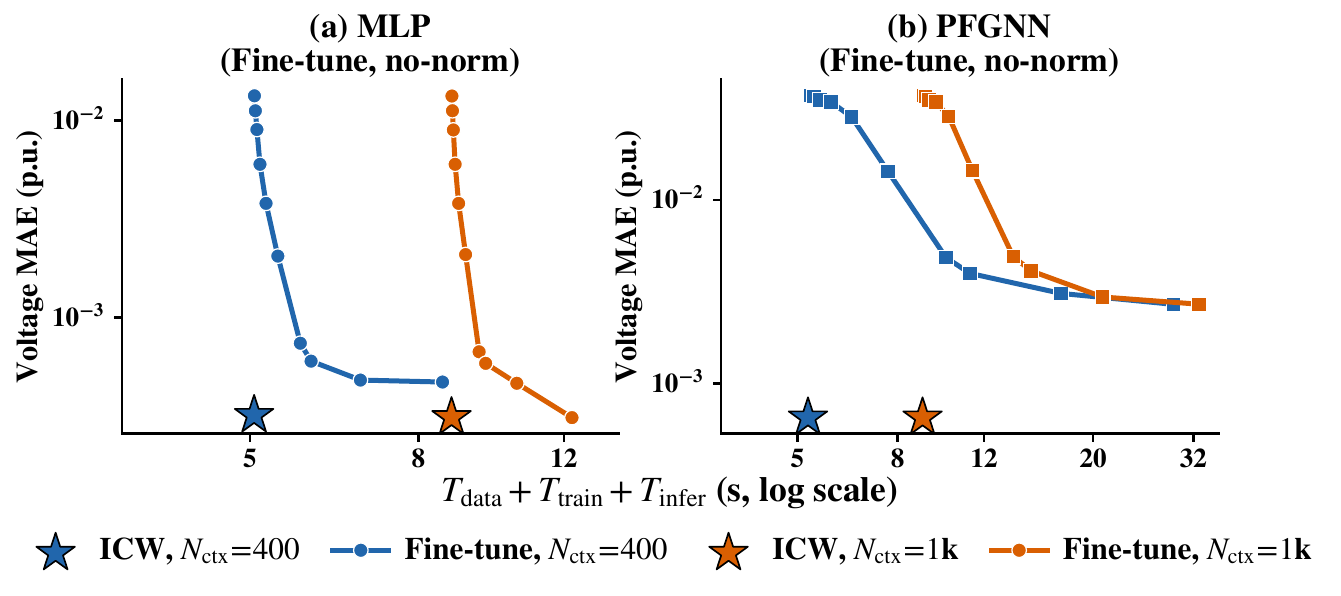}
\caption{\textbf{ICW vs.\ fine-tuning: cost and accuracy.} Voltage MAE (p.u.)
vs.\ total cost $T_\mathrm{data} + T_\mathrm{grad} + T_\mathrm{inf}$ (log
scale) for (a) MLP and (b) PFGNN, fine-tuned (no-norm) on the N-2 Most
contingency's context set at $N_\mathrm{ctx}=400$ and $N_\mathrm{ctx}=1$k. Stars
mark ICW at the same context sizes; circles/squares trace fine-tuning progress
through increasing gradient steps.}
\label{ff_figure3_combined}
\end{figure}

\textbf{Scratch training.} We next train MLP and PFGNN entirely from scratch on the same N-2 Most contingency and the same two context sizes, rather than fine-tuning from base-trained weights. The scratch-trained model starts from random initialization, so it must learn the entire input-output mapping from the context set alone rather than adjust an already-trained mapping; this isolates how much of fine-tuning's head start comes from starting near a good solution rather than from the gradient steps themselves. We run this under both no-norm and z-score preprocessing, since the two start from different places: no-norm leaves the full mean, variance, and covariance shift for gradient descent to learn unaided, while z-score removes the mean and per-dimension variance before training begins, leaving only the covariance structure to recover.

Table~\ref{tab:scratch_vs_icw_nonorm} shows what this costs under no-norm. At $100$ gradient steps, MLP's Voltage MAE is $0.1935$\,p.u., roughly $645\times$ worse than ICW's $0.0003$\,p.u., despite already costing more wall-clock time than ICW's full context-collection cost of $5.059$\,s. Scratch training only approaches ICW's accuracy after $2000$--$4000$ steps, by which point its total cost ($6.8$--$8.6$\,s for MLP, $17.2$--$29.2$\,s for PFGNN) exceeds ICW's and its MAE is still $14$--$35\times$ worse. Without the head start fine-tuning gets from pretrained weights, the model spends most of its gradient budget just recovering the mean and variance structure ICW gets for free from its closed-form statistics.

Table~\ref{tab:scratch_vs_icw_zscore} shows the same experiment under z-score. Here the story changes: at just $100$ steps MLP's Voltage MAE is already $0.0007$\,p.u., close to ICW almost immediately, since z-score has already removed most of the shift that no-norm left for gradient descent to learn. PFGNN closes its gap more slowly, sitting at $0.0011$\,p.u.\ at $100$ steps and only matching or beating ICW's $0.0007$\,p.u.\ once it reaches roughly $2000$--$4000$ steps. What remains for gradient descent under z-score is the covariance structure, the part ZCA whitening captures directly and z-score does not; ICW closes this gap with zero gradient steps, while scratch training still needs thousands to close it on this one contingency alone.

Critically, every cost above, for fine-tuning and scratch training under both no-norm and z-score, is paid for this one contingency only. ICW's surrogate is trained once on the base topology and adapts to each new contingency from context alone; fine-tuning or retraining from a single contingency's context set yields a model adapted to that contingency, with no benefit carried to the next. Repeating either procedure across 30-Bus's 38 N-1 contingencies (Section~\ref{sec:full-contingency}) would multiply its cost by $38\times$, and across case118's N-1 set by the corresponding factor, while ICW pays its context-collection cost once per contingency regardless of how many are encountered.

\begin{table}[t]
\caption{Scratch training vs.\ ICW: computation time and Voltage MAE (no-norm;
$N\!-\!2$ Most; IEEE 30-Bus; $T_{\rm grad}\!=\!0$ for ICW).}
\label{tab:scratch_vs_icw_nonorm}
\centering
\footnotesize
\setlength{\tabcolsep}{4pt}
\begin{tabular}{lrrrr}
\toprule
\textbf{Grad. Steps} & \multicolumn{2}{c}{\textbf{Time~(s)}} & \multicolumn{2}{c}{\textbf{Volt.~MAE}} \\
\midrule
\multicolumn{5}{l}{\textit{(A) $N_{\rm ctx}=400$ context samples}} \\[2pt]
\textbf{ICW (ours)} & \textbf{5.059} & \textbf{5.257} & \textbf{0.0003} & \textbf{0.0007} \\
\midrule
& \textbf{MLP} & \textbf{PFGNN} & \textbf{MLP} & \textbf{PFGNN} \\
100 & 5.143 & 5.857 & 0.1935 & 0.0834 \\
200 & 5.230 & 6.454 & 0.0782 & 0.0807 \\
400 & 5.405 & 7.649 & 0.0367 & 0.0472 \\
800 & 5.755 & 10.039 & 0.0224 & 0.0207 \\
1000 & 5.930 & 11.234 & 0.0192 & 0.0202 \\
2000 & 6.805 & 17.209 & 0.0106 & 0.0157 \\
4000 & 8.555 & 29.159 & 0.0042 & 0.0112 \\
\midrule
\multicolumn{5}{l}{\textit{(B) $N_{\rm ctx}=1000$ context samples}} \\[2pt]
\textbf{ICW (ours)} & \textbf{8.772} & \textbf{8.994} & \textbf{0.0003} & \textbf{0.0007} \\
\midrule
& \textbf{MLP} & \textbf{PFGNN} & \textbf{MLP} & \textbf{PFGNN} \\
100 & 8.861 & 9.575 & 0.1547 & 0.0808 \\
200 & 8.948 & 10.172 & 0.0665 & 0.0657 \\
400 & 9.123 & 11.367 & 0.0292 & 0.0498 \\
800 & 9.473 & 13.757 & 0.0191 & 0.0271 \\
1000 & 9.648 & 14.952 & 0.0149 & 0.0226 \\
2000 & 10.523 & 20.927 & 0.0079 & 0.0130 \\
4000 & 12.273 & 32.877 & 0.0036 & 0.0073 \\
\bottomrule
\end{tabular}
\end{table}

\begin{table}[t]
\caption{Scratch training vs.\ ICW: computation time and Voltage MAE (z-score;
$N\!-\!2$ Most; IEEE 30-Bus; $T_{\rm grad}\!=\!0$ for ICW).}
\label{tab:scratch_vs_icw_zscore}
\centering
\footnotesize
\setlength{\tabcolsep}{4pt}
\begin{tabular}{lrrrr}
\toprule
\textbf{Grad. Steps} & \multicolumn{2}{c}{\textbf{Time~(s)}} & \multicolumn{2}{c}{\textbf{Volt.~MAE}} \\
\midrule
\multicolumn{5}{l}{\textit{(A) $N_{\rm ctx}=400$ context samples}} \\[2pt]
\textbf{ICW (ours)} & \textbf{5.059} & \textbf{5.257} & \textbf{0.0003} & \textbf{0.0007} \\
\midrule
& \textbf{MLP} & \textbf{PFGNN} & \textbf{MLP} & \textbf{PFGNN} \\
100 & 5.143 & 5.857 & 0.0007 & 0.0011 \\
200 & 5.230 & 6.454 & 0.0004 & 0.0011 \\
400 & 5.405 & 7.649 & 0.0003 & 0.0011 \\
800 & 5.755 & 10.039 & 0.0002 & 0.0009 \\
1000 & 5.930 & 11.234 & 0.0002 & 0.0009 \\
2000 & 6.805 & 17.209 & 0.0001 & 0.0007 \\
4000 & 8.555 & 29.159 & 0.0001 & 0.0004 \\
\midrule
\multicolumn{5}{l}{\textit{(B) $N_{\rm ctx}=1000$ context samples}} \\[2pt]
\textbf{ICW (ours)} & \textbf{8.772} & \textbf{8.994} & \textbf{0.0003} & \textbf{0.0007} \\
\midrule
& \textbf{MLP} & \textbf{PFGNN} & \textbf{MLP} & \textbf{PFGNN} \\
100 & 8.861 & 9.575 & 0.0008 & 0.0011 \\
200 & 8.948 & 10.172 & 0.0004 & 0.0011 \\
400 & 9.123 & 11.367 & 0.0002 & 0.0011 \\
800 & 9.473 & 13.757 & 0.0001 & 0.0010 \\
1000 & 9.648 & 14.952 & 0.0001 & 0.0009 \\
2000 & 10.523 & 20.927 & 0.0001 & 0.0005 \\
4000 & 12.273 & 32.877 & 0.0001 & 0.0004 \\
\bottomrule
\end{tabular}
\end{table}

\end{document}